\documentclass[12pt]{article}
\pdfoutput=1
\usepackage{natbib}
\usepackage{multirow}

\usepackage{amsthm}
\usepackage{amssymb}
\usepackage{xcolor}
\usepackage{bbm}
\usepackage{lscape}
\usepackage{caption}
\usepackage{subcaption}
\usepackage{mathtools}
\usepackage{url}
\usepackage{float}
\usepackage{amsmath}
\usepackage{mathalfa}

\usepackage{algorithm}
\usepackage[parfill]{parskip}
\usepackage{algorithmic}
\usepackage[T1]{fontenc}
\usepackage{titling}
\usepackage[toc,page]{appendix}
\usepackage{listings}
\usepackage{geometry}
\usepackage{dsfont}
\usepackage{wrapfig}
\usepackage{graphicx}
\usepackage[english]{babel}
\newtheorem{theorem}{Theorem}[section]

\numberwithin{equation}{section}

\newcommand{\beginsupplement}{%
        \setcounter{section}{0}%
        \setcounter{table}{0}
        \renewcommand{\thetable}{S\arabic{table}}%
        \setcounter{figure}{0}
        \renewcommand{\thefigure}{S\arabic{figure}}%
        \setcounter{algorithm}{0}
        \renewcommand{\thealgorithm}{S\arabic{algorithm}}%
        \setcounter{equation}{0}
        \renewcommand{\theequation}{S\thesection.\arabic{equation}}%
     }

\theoremstyle{definition}
\newtheorem{exmp}{Example}[section]

\definecolor{ao}{rgb}{0.0, 0.5, 0.0}

\title{A multi-arm multi-stage design for trials with all pairwise testing}
\author{Peter Greenstreet$^{1,2}$, Thomas Jaki$^{3,4}$, Alun Bedding$^{5}$, Pavel Mozgunov$^{3}$
\\ \footnotesize
$^1$ Ottawa Methods Centre, Ottawa Hospital Research Institute, Ottawa, Canada \\ \footnotesize
$^2$ Department of Mathematics and Statistics, Lancaster University, Lancaster, UK \\ \footnotesize
$^3$ MRC Biostatistics Unit, University of Cambridge, Cambridge, UK \\ \footnotesize
$^4$ University of Regensburg, Regensburg, Germany \\ \footnotesize
$^5$ Alun Bedding Coaching \& Consulting Ltd, Bury St Edmunds, UK
}
\date{}
\begin{document}
\maketitle
\begin{abstract}
Multi-arm multi-stage (MAMS) trials have gained popularity to enhance the efficiency of clinical trials, potentially reducing both duration and costs. This paper focuses on designing MAMS trials where no control treatment exists. This can arise when multiple standard treatments are already established or no treatment is available for a severe disease, making it unethical to withhold a potentially helpful option. The proposed design incorporates interim analyses to allow early termination of notably worst treatments and stops the trial entirely if all remaining treatments are performing similarly. The proposed design controls the familywise error rate (FWER) for all pairwise comparisons and provides the conditions guaranteeing FWER control in the strong sense. The FWER and power are used to calculate both the stopping boundaries and the sample size required. Analytic solutions to compute the expected sample size are also derived. A trial motivated by a study conducted in sepsis, where there was no control treatment, is shown. The multi-arm multi-stage all pairwise (MAMSAP) design proposed here is compared to multiple different approaches. For the trial studied, the proposed method yields the lowest required maximum and expected sample size when controlling the FWER and power at the desired levels.  
\end{abstract}
\section{Introduction}
Multi-arm multi-stage trials have become increasingly popular due to their potential to reduce the duration and cost of clinical trials \citep{StallardNigel2020EADf, lee2021}. Multi-arm studies can have multiple potential benefits including: shared trial infrastructure and control arm; 
less administrative and logistical effort than setting up separate trials and enhanced recruitment \citep{BurnettThomas2020Aeta, MeurerWilliamJ2012ACTA}. Interim analyses can greatly improve the efficiency of a clinical trial and help avoid unnecessary exposure of participants to ineffective or harmful treatments, while also conserving patients that could be redirected to more promising treatments \citep{PocockStuartJ.1977GSMi,WasonJames2016Srfm}. This results in useful therapies potentially being identified faster while reducing cost and time \citep{cohen2015adding}. 
Traditionally multi-arm multi-stage (MAMS) trials involve comparing the active treatments to a \textit{common control} treatment at predefined interim stages \citep{WasonJamesM.S.2012Odom, RoystonPatrick2003Ndfm, greenstreet2023change}. \cite{MagirrD.2012AgDt} extended the multi-arm setting with common control treatment of \cite{DunnettCharlesW1955AMCP} to allow for a MAMS design in which the type I error of the entire trial is controlled.
\par 
In this manuscript we will focus on designing multi-arm multi-stage trials where no control treatment is available. Specifically, we are extending the work of \cite{TukeyJohnW.1949CIMi} to allow for interim analyses while still controlling the type I error of the entire trial. There are several applied settings where control treatments are absent, for instance when multiple treatments are already established as the standard of care for a condition and the objective of the trial is to determine if any treatment(s) is/are superior or inferior to any of the others \citep{BriffaTom2021Ncet}. Such investigations are particularly important, as in many medical specialties, less than 20\% of recommendations in contemporary clinical practice guidelines are supported by high quality evidence \citep{ CaliffRobertM2016TEGt}. 
Another situation where such trials are useful is where no treatment exists for a specific severe disease in a given population exists so that it would be unethical to give patients a placebo and thereby withholding a potentially beneficial treatment. There may be no treatment currently used due to a lack of resources to use the accepted standard of care, or if it is an emerging infectious disease, so no standard of care has been established \cite{whitehead2016trial, MagaretAmalia2016Doam}.
\par
\cite{MagaretAmalia2016Doam} propose an approach for how one can conduct all pairwise comparisons for a multi-arm study with no control treatment in sepsis where the trial has interim analyses. This trial was motivated by the Ebola outbreak \citep{whitehead2016trial}. When conducting pairwise comparisons, all the null hypotheses, that two treatments are equal, are tested for every pair of treatments within the multi-arm study. In this proposal, a treatment is dropped, at an interim analysis, if it is found to be statistically significantly worse then at least one other treatment in the trial and if all remaining treatments are found to be similar, then the trial stops. 
 In \cite{MagaretAmalia2016Doam} the calculations of the rules to drop treatments or stop the trial early were done using a simulation based approach which did not guarantee the type I error of the entire trial. This work was then considered in \cite{WhiteheadJohn2020Eote} which proposed a different design based on using the double triangular stopping rules \citep{WhiteheadJ.1997TDaA, whitehead1990double} 
 to define when treatments would stop in the trial. In \cite{WhiteheadJohn2020Eote} the boundaries were set to control the type I error for \textit{each pairwise comparison} and not adjusted to account for the multiplicity of the design. Therefore control of the power and overall type I error of the trial were not guaranteed. 
\par
An alternative to conducting an all pairwise approach is to use a screened selection design such as the one discussed in \cite{WuJianrong2022Tssd}. In this design there is no control treatment and the treatments are ranked against each other and decisions are made based on a drop the loser design or pick-the-winner design \citep{HillsRobertK.2011AoaP}. 
For this type of design it is not possible to control the probability of wrongly declaring one treatment better than another when in fact they have equal treatment effect. Consequently \cite{WuJianrong2022Tssd} propose using such a design for phase II screening. Therefore it is less applicable for late phase trials which are the focus of this work.  
\par
In this work, all pairwise comparisons are made to compare the multiple treatment arms to one another and interim analyses allow for early termination of treatments found to be inferior to others and can lead to the early termination of the entire trial if all remaining treatments are deemed similar. We focus on guaranteeing family-wise error rate (FWER) control, where FWER is the probability of rejecting any true null hypotheses across the entire trial. The FWER is considered a robust and strong type of error control in multi-arm trials \citep{WasonJames2016Srfm} and in certain scenarios, it is recommended or even required by regulatory authorities \citep{WasonJamesMS2014Cfmi}.  This work also presents an analytical approach to finding the required sample size which guarantees the desirable statistical power.
\par
The upcoming section will formally introduce the motivating example and give its key characteristics. This motivating example is then used throughout the following methodology section to introduce the proposed multi-arm multi-stage all pairwise (MAMSAP) design. In Section \ref{Sec:Methodology}, the FWER is formally defined and design consideration for FWER control in the strong sense are given, along with the methodology for calculating power of the trial. 
The design for the motivating example using the MAMSAP design is presented in Section \ref{Sec:Motivating} and is compared to alternative approaches before the paper concludes with a discussion.

\section{Motivating example}
\label{Sec:MotivatingExample}
We are motivated by the design for a trial in sepsis as discussed in \cite{MagaretAmalia2016Doam} and \cite{WhiteheadJohn2020Eote}. Guidelines exist on how to treat patients with sepsis \citep{DunserMartin2012Rfsm}, however there is no current standard of care treatment, so a multi-arm all pairwise trial was suggested due to the high mortality rate of over 55\% with current practice \citep{whitehead2016trial}. In both \cite{MagaretAmalia2016Doam} and \cite{WhiteheadJohn2020Eote} the binary outcome of mortality of patients after 28 days is used as the primary endpoint.
\par
%
Motivated by this trial, we focus on a trial with 4 arms and 3 stages per treatment arm with equal numbers of patients per treatment per stage.  We use the same trial configuration as discussed in \cite{WhiteheadJohn2020Eote} of a normal approximation of the binary treatment effect difference with a clinically relevant effect ($\theta'$) of $\log(1.5)$, which equals a 50\% improvement in survival.  

\section{Methodology}
\label{Sec:Methodology}

\subsection{Setting}
\label{Sec:Setting}
Let $K$ be the number of treatments for the trial with the primary outcome measured for each patient being assumed to be independent. The treatment effect of the $K$ experimental treatments are  $\psi_1, \hdots, \psi_K$. 
Let $H_{k,k^\star}$ define the null hypothesis for treatment $k$ with treatment $k^\star$, where $k \neq k^\star$ and $k,k^\star = 1,\hdots, K$.    
The set of null hypotheses for an all pairwise comparison trial are
\begin{align*}
H_{1,2}: \psi_1=\psi_2,\hdots, & H_{1,K}:\psi_1=\psi_K, \;  H_{2,3}:\psi_2=\psi_3, \hdots, H_{K-1,K}:\psi_{K-1}=\psi_{K}.
\end{align*}
The number of null hypotheses equals $\eta=\sum^{K-1}_{k=1} k = {K \choose 2}$ 
The global null hypothesis corresponds to $\psi_1=\psi_2=\hdots=\psi_K$.
\par
The null hypotheses are tested at a maximum of $J$ stages with there being a maximum of $J-1$ interim analyses, with an analysis taking place at the end of each stage. At each analysis the null hypothesis for each pairwise comparison in the trial is tested until at least one of the treatments in that hypothesis being tested has stopped, then the corresponding null hypothesis is not tested for the remaining trial.  
The null hypothesis at stage $j$ for treatment $k$ and $k^\star$ is tested using the test statistics 
\begin{equation*}
Z_{(k,k^\star),j} =\frac{\bar{\psi}_{k,j}- \bar{\psi}_{k^\star,j}}{\sqrt{V_{(k,k^\star),j}}},
\end{equation*}
where $\bar{\psi}_{k,j}$ and $\bar{\psi}_{k^\star,j}$ are the treatment effects of the observed patients on that given treatments $k,k^\star$, up to stage $j = 1, \hdots, J$ and $V_{(k,k^\star),j}$ is the variance of the observed difference in treatment effects. For the motivating example  $V_{(k,k^\star),j}$ is given in Example \ref{exmp1}. It is assumed that $Z_{(k,k^\star),j}$ follows a normal distribution $Z_{(k,k^\star),j} \sim (\frac{\bar{\psi}_{k,j}- \bar{\psi}_{k^\star,j}}{\sqrt{V_{(k,k^\star),j}}},1)$.
Note that these are the same test statistics as used for the Tukey test \citep{TukeyJohnW.1949CIMi, KramerClydeYoung1956EoMR}. 
\par
The decision-making for the trial is made using \textit{outer} upper and lower stopping boundaries and \textit{inner} upper and lower stopping boundaries.
The outer boundaries are used first at a given stage to test if there is a statistically significant difference between two treatments, so if there is then the inferior treatment is dropped from the trial. The outer upper boundaries are denoted as $U=(u_{1},\hdots,u_{J})$ and the outer lower boundaries are denoted as $L=(-u_{1},\hdots,-u_{J})$, where $u_j$ is the upper outer boundary at stage $j$, $j=1,\hdots,J$. The outer upper and lower boundaries are symmetric as significant differences in both directions are equally important. If $Z_{(k,k^\star),j}>u_j$ then treatment $k$ is declared superior to treatment $k^\star$ and treatment $k^\star$ is dropped from the trial. If $Z_{(k,k^\star),j}<-u_j$ then treatment $k^\star$ is declared superior to treatment $k$ and treatment $k$ is dropped from the trial.
\par
The inner boundaries are then used, at a given stage to test if all the remaining treatments are similar enough to stop the trial early. The inner upper boundaries and inner lower boundaries are also symmetric and denoted as $U^\star=(u^\star_{1},\hdots,u^\star_{J})$ and $L^\star=(-u^\star_{1},\hdots,-u^\star_{J})$, respectively, where $u^\star_j$ is the upper inner boundary at stage $j$. For stages where one is not testing if the remaining treatments are similar enough to stop the trial early then $u^\star_j=0$. For the inner upper and lower boundaries if  $-u^\star_j < Z_{(k,k^\star),j}< u^\star_j$ for all treatments that have not been dropped by stage j, then the trial stops with the conclusion that the remaining treatments are similar. If at least 2 treatments, $k,k^\star$, exist that have not been dropped by stage $j$ and $-u_j<Z_{(k,k^\star),j}<-u^\star_j$ or $u^\star_j<Z_{(k,k^\star),j}<u_j$ then all treatments that have not been dropped by stage $j$ continue to the next stage.

\par
In this work both binding and non-binding boundaries will be considered when calculating the FWER. In the context of this design binding rules require trial termination if all treatments are found to be similar at a given stage, while non-binding rules grant the trial team the flexibility to decide whether to continue or stop the trial \citep{li2020optimality, BretzFrank2009Adfc}. Binding and non-binding boundaries both require that a treatment is dropped if it is found inferior to another treatment. In other words, the outer bounds are always binding while the inner bounds will be considered to be binding or non-binding. 
\begin{exmp}
\label{exmp1}
For the motivating example using the normal approximation, the variance of the observed difference in treatment effect equals $V_{(k,k^\star),j}=(n_{k,j}^{-1}+ n_{k^\star,j}^{-1})^{-1}$ where $n_{k,j}$ denotes the number of patients recruited to treatment $k$ by the end of stage $j$. Similarly we define $r_{k,j}$ as the ratio of patients recruited to treatment $k$ by the end of stage $j$ compared to the number of patients recruited to treatment 1 by the end of stage 1, so that $r_{1,1}=1$. The realized sample size of a trial is denoted by $N$ with the maximum planned sample size being $\max(N)= \sum_{k=1}^K n_{k,J}$ where $K$ is the number of treatments in the trial and where $J$ is the maximum number of analyses for the trial.
\par
Based on the motivating example the boundary shape when using double triangular boundaries \citep{WhiteheadJ.1997TDaA} are given in Figure~\ref{fig:Triplot} to control the FWER at 5\% of the trial for binding boundaries. Shown in this figure are both the outer and inner boundaries, as well as the different areas for which each test statistic could fall. The horizontal lines represent the area that one would reject the null hypothesis. The solid area being where the hypothesis is unable to be rejected but the hypothesis will continue being studied. If all remaining test statistics are in the vertical lined area then the trial stops for all remaining treatments being similar. 

\begin{figure}[h]
  \centering
  \includegraphics[width=.70\linewidth,trim= 0 0.5cm 0 2cm, clip]{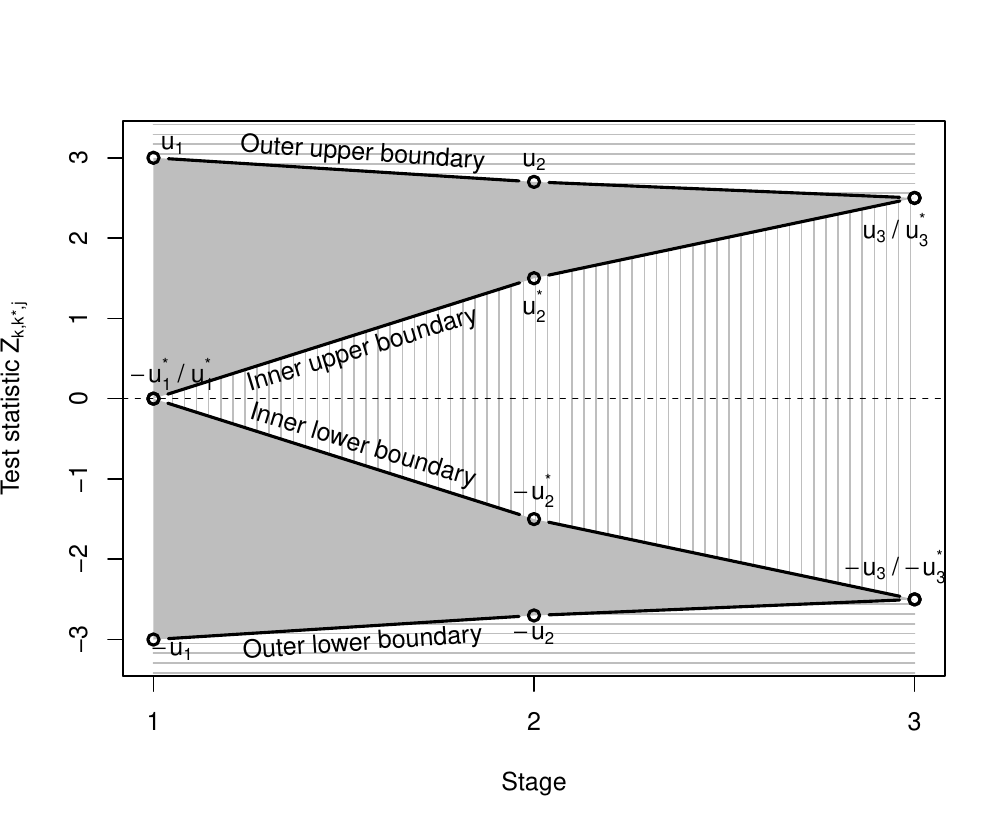}  

\caption{ \small The boundary shape when using the binding double triangular boundaries for the 3 stage motivating example.} 
\label{fig:Triplot}
\end{figure} 
\hfill $\blacksquare$
\end{exmp}
\subsection{Familywise error rate (FWER)}
\label{Sec:FWER}
In an all pairwise trial the type I error for each comparison is the probability that the null hypothesis for that comparison is wrongly rejected at any stage during the trial, when the null hypothesis is true. The FWER is the probability of making at least one type I error across all the comparisons at any stage of the trial. The FWER in the strong sense is defined as:
\small
\begin{equation*}
P(\text{reject at least one true } H_{k,k^\star} \text{  under any null configuation}, k,k^\star=1,\hdots, K \text{  given } k \neq k^\star) \leq \alpha,
\end{equation*}
\normalsize
where $\alpha$ is the desired level of control. Control of the FWER in the strong sense means that the FWER is controlled under any null configuration of treatment effects. Control of the FWER in the weak sense means that the FWER is only guaranteed to be controlled under the global null configuration \citep{WasonJamesMS2014Cfmi}.  
To calculate FWER we define the following events
$b_{(k,k^\star),j}=\{-u_{j} <Z_{(k,k^\star),j}< u_{j}\}$, $c_{(k,k^\star),j}=\{-u^\star_{j} <Z_{(k,k^\star),j}< u^\star_{j}\},$
where $b_{(k,k^\star),j}$ is the event that the test statistic testing treatment $k$ against treatment $k^\star$ is within the outer boundaries at stage $j$ and $c_{(k,k^\star),j}$ is the event the test statistic is within the inner boundaries at stage $j$.
\par
We define $T_{\beta,j}$ as the set of indices of true null hypotheses being tested at stage $j = 1,\hdots, J$, where $\beta$ is used to denote that these are the null hypotheses being tested at the start of a given stage. Therefore $T_{\beta,j}$ is given before any treatments are dropped for inferiority at given stage $j$. We define $T_{\gamma,j}$ as the set of indices of hypotheses being tested after dropping any treatments found to be inferior to any other treatment by the end of stage $j$, so $\gamma$ is used to denote that these are the null hypotheses after dropping treatments at stage $j$. Therefore $T_{\gamma,j}$ is given after any treatments are dropped at stage $j$ but also includes any remaining null hypotheses even if they are not true null hypotheses as all remaining test statistics need to be in the inner boundaries for the trial to stop.  We define at the final stage $J$ that $T_{\gamma,J}=T_{\beta,J}$ as at the final stage, only the true null hypotheses being tested will effect the type I error as the trial will end at this given stage.
We denote the family of sets of $T_{\beta,j}$, for every $j=1,...,J$, as $\mathbf{T_\beta}=\{T_{\beta,1} ,\hdots, T_{\beta,J} \}$ and similarly denote the family of sets of $T_{\gamma,j}$, for every $j=1,...,J$, as  $\mathbf{T_\gamma}=\{T_{\gamma,1} ,\hdots, T_{\gamma,J} \}$. In addition we define $C_{T_{\gamma,j},j}$ as the event that all the test statistics for stage $j$ in the set indexed in $T_{\gamma,j}$ are within the inner boundaries and $B_{T_\beta,j}$ as the event that all the test statistics for stage $j$ in the set indexed in $T_{\beta,j}$ are within the outer boundaries, so,
$
B_{T_{\beta,j},j}=\bigcap_{h \in T_{\beta,j} } b_{(h),j}$ and
$C_{T_{\gamma,j},j}=\bigcap_{h \in T_{\gamma,j} } c_{(h),j}$.

\par
Additionally we define $D_{T_{\beta,j},T_{\gamma,j},j}=B_{T_{\beta,j},j} / C_{T_{\gamma,j},j}$, so $D_{T_{\beta,j},T_{\gamma,j},j}$ are the events that all the test statistics testing the true null hypotheses are within the outer boundaries, but at least one of the test statistics still being tested, at the end of stage $j$, is outside the inner boundaries.  
\begin{exmp}
For the motivating example assume that at the beginning of testing at stage 2 treatments 1, 2 and 3 are still being tested and $\psi_1=\psi_2 \neq \psi_3$. Then $T_{\beta,2}=\{(1,2)\}$. If no treatments are found inferior to any other treatments at this stage, i.e all test statistics are within the outer boundaries, then $T_{\gamma,2}=\{(1,2),(1,3),(2,3)\}$. Also $C_{T_{\gamma,2},2}= c_{(1,2),2} \cap c_{(1,3),2} \cap c_{(2,3),2}$ and $D_{T_{\beta,2},T_{\gamma,2},2}=  (b_{(1,2),2}) / (c_{(1,2),2} \cap c_{(1,3),2} \cap c_{(2,3)})$. \hfill $\blacksquare$
\label{Example:Imagine}
\end{exmp}

\subsubsection{FWER for non-binding inner stopping rules}
The FWER is at its greatest if one does not account for the inner boundaries stopping rules i.e. non-binding. Therefore the event, $R'_{\mathbf{T_\beta}}$, where no true null hypotheses are rejected under any given $ \mathbf{T_\beta}=\{T_{\beta,1} ,\hdots, T_{\beta,J} \}$ for a trial with J stages when using non-binding stopping rules equals
$R'_{\mathbf{T_\beta}}= \bigcap^J_{j=1} B_{T_{\beta,j},j}$.  
Under the global null hypothesis $T_{\beta,j}=G$, where $G$ is the set of all the indices for each hypothesis $G=\{(1,2),\hdots, (K-1,K)\}$. The event that no true null hypotheses are rejected simplifies to
$
R'_{\mathbf{G}}= \bigcap^J_{j=1} B_{G,j},  
$
where $\mathbf{G}$ is length $J$ and $\mathbf{G}=\{ G,\hdots, G \}$. The FWER under the global null therefore equals $1-P(R'_{\mathbf{G}})$.
\par
\begin{theorem}
The probability of rejecting at least one true null hypotheses is maximized under the global null hypothesis when non-binding stopping rules are used.
\label{theorem:nonbind}
\end{theorem}
\begin{proof}
We begin by defining $T_\beta^\star$ as $T_{\beta,1}$ therefore $T_{\beta,j} \subseteq T_\beta^\star$ and we define $\mathbf{T_\beta^\star}=\{T_\beta^\star,\hdots, T_\beta^\star \}$, so
\begin{equation*}
R'_{\mathbf{T_\beta}}=\bigcap^J_{j=1} B_{T_{\beta,j},j} = \bigcap^J_{j=1}\bigcap_{h \in T_{\beta,j}} b_{h,j}   \supseteq  \bigcap^J_{j=1}\bigcap_{h \in T^\star_{\beta}} b_{h,j} =\bigcap^J_{j=1} B_{T^\star_{\beta},j}=R'_{\mathbf{T^\star_\beta}}.
\end{equation*} 
Then as $T_\beta^\star \subseteq G$,
\begin{equation*}
R'_{\mathbf{T^\star_\beta}}=\bigcap^J_{j=1} B_{T^\star_{\beta},j}=\bigcap^J_{j=1}\bigcap_{h \in T^\star_{\beta}} b_{h,j} \supseteq  \bigcap^J_{j=1}\bigcap_{h \in G} b_{h,j} =\bigcap^J_{j=1} B_{G,j}=R'_{\mathbf{G}}.
\end{equation*} 
Therefore $
1-P(R'_{\mathbf{T_\beta}}) \leq  1-P(R'_{\mathbf{G}}).
$
\end{proof}
Theorem \ref{theorem:nonbind} shows that for the non-binding stopping boundaries, the FWER is maximised under the global null hypothesis, so that if FWER control is at level $\alpha$ under the global null hypothesis then this implies FWER control in the strong sense at level $\alpha$. 
To compute the FWER under the global null hypothesis one needs to calculate $P(R'_{\mathbf{G}})$. This can use the multivariate normal distribution similar to as done in \cite{MagirrD.2012AgDt, GreenstreetPeter2021Ammp}. The exact details on how the probabilities can be computed for non-binding and binding boundaries are provided in the Supporting Information Section 1.

\subsubsection{FWER for binding stopping rules}
When using binding boundaries one can now use the fact that the trial is guaranteed to stop early if all the test statistics of the remaining treatments are within the inner boundaries, along with being able to drop treatments found inferior to other treatments.  When using binding stopping rules the event that no true null hypotheses are rejected under any given set of indices $\mathbf{T_\beta}$ and $\mathbf{T_\gamma}$ ($R_{\mathbf{T_\beta},\mathbf{T_\gamma}}$) equals
\begin{equation*}
R_{\mathbf{T_\beta},\mathbf{T_\gamma}}=\bigcup^{J}_{j=1} \bigg{(} [B_{T_{\beta,j},j} \cap C_{T_{\gamma,j},j}] \cap \bigcap^{j-1}_{i=1} (D_{T_{B,i},T_{\gamma,i},i}) \bigg{)}.  
\end{equation*}
The FWER for given $\mathbf{T_\beta}$ and $\mathbf{T_\gamma}$ is therefore $1-P(R_{\mathbf{T_\beta},\mathbf{T_\gamma}})$.
Similar to the case of the non-binding boundaries, this equation can also be simplified when under the global null. Now $T_{\gamma,j}=T_{\beta,j}=G$ as none of the test statistics can be stopped early from being found inferior compared to another treatment without this being a type I error. One can now use the fact that $C_{G,j}\subseteq B_{G,j}$, and define $D_{G,G,j}=B_{G,j}-C_{G,j}$, as the difference of two sets where the latter set is a subset of the former. Therefore, when using binding stopping rules the event that no treatments are found superior to any other treatment under the global null equals
\begin{equation*}
R_{\mathbf{G},\mathbf{G}}= \bigcup^J_{j=1} \bigg{(} C_{G,j} \cap \bigcap^{j-1}_{i=1} (D_{G,G,i}) \bigg{)}=\bigcup^J_{j=1} \bigg{(} C_{G,j} \cap \bigcap^{j-1}_{i=1} (B_{G,i}-C_{G,i}) \bigg{)}.  
\end{equation*}
The FWER under the global null equals $1-P(R_{\mathbf{G},\mathbf{G}})$.
For binding boundaries, however, it is not always true that controlling the FWER under the global null will result in strong control of the FWER. A simple example of this is given in the Supporting Information Section 2, where it is shown that if one has very wide inner boundaries then under the global null hypothesis the final outer boundaries can be very tight leading to inflated FWER when not under the global null hypothesis.  
\par
While one can not ensure FWER control in the strong sense being implied by control under the global null hypothesis, one can determine if it is for the given boundaries.  This test involves comparing the FWER under the global null to a finite set of alternative configurations assuming the use of non-binding boundaries. The finite set of alternatives is a reduced set of all possible indices $T^\star_\beta$ excluding the empty set and full set. The complete set of indices $T^\star_\beta$ is defined as $\textbf{S}$ which equals  $\textbf{S}=\{S_1,\hdots,S_I\}$  where each $S_i$ is a unique $T^\star_\beta$ for all $i=1,\hdots,I$, where the number of sets of null hypotheses $I$ equals the Bell number minus 2 \citep{BellE.T.1938TIEI}, $Bell(K)-2$. Additionally we define $\mathbf{S'}=\{S'_1,\hdots,S'_{I'}\}$ with $\mathbf{S'}\subseteq \mathbf{S}$ such that $S_i \subseteq \{S'_1,\hdots,S'_{I'}\}$ for all $i=1,\hdots I$. The number of sets, $I'$, in $\mathbf{S'}$ is the Stirling number of the second kind  \citep{graham1989concrete}, $Stirling(K,2)$ as demonstrated in Example \ref{Exmp:Sprime}.  
\par
Theorem \ref{theorem:FWERshow} uses the finite set $\mathbf{S'}$ to test if the FWER is controlled in the strong sense under the global null hypothesis for binding boundaries. One tests every possible set of null hypotheses excluding the global null hypothesis for non-binding boundaries. As shown below, if every possible set can be shown to have lower FWER than the FWER under the global null hypothesis for the binding boundaries then the FWER is controlled in the strong sense under the global null hypothesis for the given binding boundary. This is based on the fact that the FWER for binding boundaries is less than that of non-binding boundaries as proven in Theorem \ref{theorem:FWERshow}.

\begin{theorem}
If the FWER for binding stopping rules under the global null hypothesis is greater than or equal to $1-P\bigg{(}\bigcap^J_{j=1} B_{S'_{i'},j} \bigg{)}$ for all $S'_{i'} \in \mathbf{S'}$ then the FWER for the binding boundaries is controlled in the strong sense under the global null hypothesis.
\label{theorem:FWERshow}
\end{theorem}
\begin{proof}
The event of not rejecting any set of true null hypotheses for any $\mathbf{T_{B}}$ and $\mathbf{T_{\gamma}}$ equals:
\begin{align}
R_{\mathbf{T_\beta},\mathbf{T_\gamma}}=\bigcup^{J}_{j=1} \bigg{(} [B_{T_{\beta,j},j} \cap C_{T_{\gamma,j},j}] \cap \bigcap^{j-1}_{i=1} (D_{T_{\beta,i},T_{\gamma,i},i}) \bigg{)} \supseteq \bigcap^J_{j=1} B_{T_{\beta,j},j} \supseteq \bigcap^J_{j=1} B_{T^\star_{\beta},j}
\label{Equ:Prooflessthan}
\end{align}
For $T^\star_{\beta}$ to be a set of true null hypotheses implies $T^\star_{\beta} \in \mathbf{S}$ so that $T^\star_{\beta} \subseteq \{S'_1,\hdots,S'_{I'}\}$. Therefore if
\begin{align}
1-P\bigg{(}\bigcap^J_{j=1} B_{S'_{i'},j} \bigg{)} \leq 1-P\bigg{(}\bigcap^J_{j=1} R_{\mathbf{G},\mathbf{G}} \bigg{)}
\label{Equ:Proofmain}
\end{align}
for all $i'=1,\hdots, I'$ it follows that for any set of possible true null hypotheses, 
\begin{align*}
\bigcap^J_{j=1} B_{T^\star_{B},j} = \bigcap^J_{j=1}\bigcap_{h \in T_{\beta,j}} b_{h,j}   \supseteq \bigcap^J_{j=1}\bigcap_{h \in S_i} b_{h,j} = \bigcap^J_{j=1} B_{T_{S_{i} },j} \supseteq \bigcap^J_{j=1}\bigcap_{h \in S'_{i'} } b_{h,j} = \bigcap^J_{j=1} B_{T_{S'_{i'} },j}
\end{align*}
as $S_i \subseteq S'_{i'}$ for some $i \in 1,\hdots, I$ and some $i' \in 1,\hdots, I'$. If Equation \eqref{Equ:Proofmain} holds for all $T^\star_{\beta} \in \mathbf{S}$ then
$
1-P\bigg{(}R_{\mathbf{T_\beta},\mathbf{T_\gamma}} \bigg{)} \leq 1-P\bigg{(} R_{\mathbf{G},\mathbf{G}} \bigg{)}.
$
\end{proof}
To check if the boundaries that control the FWER under the global null hypothesis also control the FWER in the strong sense then one needs to check that for the chosen boundaries that $P(R_{\mathbf{G},\mathbf{G}}) \leq P\bigg{(}\bigcap^J_{j=1} B_{S'_{i'},j} \bigg{)}$ for all $S'_{i'}$ in $\mathbf{S'}$ in Equation \eqref{eq:exampleSprime}. When calculating this one can use the fact $ P\bigg{(}\bigcap^J_{j=1} B_{S'_{i'},j} \bigg{)}= P(R'_{\mathbf{S'_{i'}}})$ where $\mathbf{S'_{i'}}$ is of length $J$ and $\mathbf{S'_{i'}}=\{S'_{i'}, S'_{i'}, \hdots, S'_{i'}\}$ as demonstrated in Example \ref{Exmp:Sprime}. Therefore it can be calculated in a similar manner to $P(R'_{\mathbf{G}})$ as described in the Supporting Information Section 1.
\par
The order of treatments has no effect on the calculation of $P\bigg{(}\bigcap^J_{j=1} B_{S'_{i'},j} \bigg{)}$ provided that the number of elements are the same and so is the sample size for each treatment. Therefore when there is equal sample size per treatment at each stage $\mathbf{S'}$ can be further reduced to be of length $\lceil (K-1)/2 \rceil$ as shown at the end of Example \ref{Exmp:Sprime}.
\par
If the requirements of Theorem \ref{theorem:FWERshow} are not met, one can guarantee control of FWER by determining the design using non-binding boundaries under the global null hypothesis. By using Theorem \ref{theorem:FWERbindvsnon} these boundaries will be conservative for binding boundaries but guarantee strong control of the FWER. 
\begin{theorem}
The FWER is greater or equal for the non-binding boundaries compared to the binding boundaries for a given $\mathbf{T_\beta}$.
\label{theorem:FWERbindvsnon}
\end{theorem}
\begin{proof}
From Equation \eqref{Equ:Prooflessthan},
$
1-P\bigg{(}R_{\mathbf{T_\beta},\mathbf{T_\gamma}} \bigg{)} \leq 1-P\bigg{(} R'_\mathbf{\mathbf{T_\beta}} \bigg{)}.
$
\end{proof}

In the Supporting Information Section 1 the equations to calculate the FWER under the global null hypothesis, for both binding and non-binding boundaries, along with how to calculate the probabilities required for Theorem \ref{theorem:FWERshow} are given.

\begin{exmp}

Consider the motivating example of 4 arms and 3 stages using the double triangular boundaries. The binding bounds can be found under the global null hypothesis as given in Section \ref{Sec:Motivating} to control the FWER at 5\%, so that $P(R_{\mathbf{G},\mathbf{G}})=0.95$. Next one finds $P\bigg{(}\bigcap^J_{j=1} B_{S'_{i'},j} \bigg{)}=P(R'_{\mathbf{S'_{i'}}})$  for the complete set of $\mathbf{S'}$ as is shown in Table \ref{tab:SprimeforME} where $\mathbf{S'}$ has length of $Stirling(4,2)= 7$ and is 
\begin{align}
\begin{split}
\mathbf{S'}=& \Bigg{ \{ } \{(1,2),(1,3),(2,3)\}, \{(1,2),(1,4),(2,4)\}, \{(1,3),(1,4),(3,4)\},  
 \\
&  \{(2,3),(2,4),(3,4)\}, \{(1,2),(3,4)\}, \{(1,3),(2,4)\}, \{(1,4),(2,3)\}   \Bigg{ \} }.
\label{eq:exampleSprime}
\end{split}
\end{align}

\begin{table}[]
\centering
 \caption{ \small The value of $P\bigg{(}\bigcap^J_{j=1} B_{S'_{i'},j} \bigg{)}$, for given set of $S'_i \in \mathbf{S'}$ as given in Equation \eqref{eq:exampleSprime}.}
\resizebox{0.75\textwidth}{!}{\begin{tabular}{c|c c c c c c c}
& $S'_1$ & $S'_2$ & $S'_3$ & $S'_4$ & $S'_5$ & $S'_6$ & $S'_7$ 
\\
\hline
$P\bigg{(}\bigcap^J_{j=1} B_{S'_{i'},j} \bigg{)}$ & 0.972 & 0.972 & 0.972 & 0.972 & 0.979 & 0.979 & 0.979
\end{tabular}}
\label{tab:SprimeforME} 
\end{table}

Since for the motivating example $P\bigg{(}\bigcap^J_{j=1} B_{S'_{i'},j} \bigg{)}$ is  greater than $P(R_{\mathbf{G},\mathbf{G}}) =0.95$ for all $S'_{i'} \in \mathbf{S'}$, so the FWER is controlled in the strong sense. Additionally as the motivating example has equal sample size per treatment at each stage, $\mathbf{S'}$ can be further reduced to 
$
\mathbf{S'}= \Bigg{ \{ } \{(1,2),(1,3),(2,3)\}, \{(1,2),(3,4)\}  \Bigg{ \} },
$
so reducing the number of calculations. 

\hfill $\blacksquare$
\label{Exmp:Sprime} 
\end{exmp}

\subsection{Power}
\label{Sec:Power}
The power of the trial is the probability that a treatment with the clinically relevant effect is found. Similar to the definition of power under the least favourable configuration (LFC) in the MAMS case with a control treatment \citep{MagirrD.2012AgDt} we define power under the LFC as the probability that treatment $k'$ is the only treatment left by the end of the trial, given $\psi_1=\psi_2=\hdots =\psi_{k'-1}=\psi_{k'}-\theta'=\psi_{k'+1}=\hdots = \psi_K$, where $\theta'$ is the clinically relevant effect. The sample size of the trial is found to ensure that the power under the LFC is greater than $1-\beta$, where $1-\beta$ is the pre-defined level of power desired. There are multiple ways in which treatment $k'$ can become the successful treatment in the trial.
When calculating the power under the LFC one must sum over all possible configurations that end in only the clinically relevant treatment being found. 
The power under the LFC therefore equals
\begin{align}
\sum_{\Omega_{p,y} \in \Omega_{p,1} \hdots \Omega_{p,Y}}
\int^{\omega_{u,1}(t_{(1,2),1,y})}_{\omega_{l,1}(t_{(1,2),1,y})} \hdots \int^{\omega_{u,J}(t_{(K-1,K),J,y})}_{\omega_{l,J}(t_{(K-1,K),J,y})} \phi(\mathbf{z},\boldsymbol{\theta},\Sigma)  d\mathbf{z},
\label{Eq:power}
\end{align} 
where $\phi(\mathbf{z},\boldsymbol{\mu},\Sigma)$ is the probability density function of a multivariate normal distribution with mean $\boldsymbol{\mu}$ and covariance matrix $\Sigma$. Here
\begin{align*} 
\boldsymbol{\theta} =\bigg{(}\frac{\psi_{1}- \psi_{2}}{\sqrt{V_{(1,2),1}}},& \hdots, \frac{\psi_{K-1}- \psi_{K}}{\sqrt{V_{(K-1,K),1}}},\hdots \frac{\psi_{1}- \psi_{2}}{\sqrt{V_{(1,2),J}}},\hdots, \frac{\psi_{K-1}- \psi_{K}}{\sqrt{V_{(K-1,K),J}}} \bigg{)},
\end{align*}
with $\psi_1=\psi_2=\psi_{k'-1} =\psi_{k'}-\theta'=\psi_{k'+1},\hdots =\psi_K$ and $\Sigma$ is defined in the Supporting Information Section 1. Here $\Omega_{p,y}$ defines the upper and lower boundaries for a possible configuration which results in only the clinically relevant treatment being found, where $Y$ is the total number of possible configurations and $y=1,\hdots,Y$. Each $\Omega_{p,y}$ is a list of upper and lower boundaries for each test statistic and each stage, so $\Omega_{p,y}=\{\omega_{1}(t_{(1,1),1,y}), \hdots, \omega_{J}(t_{(K-1,K),J,y})\}$, with $\omega_{j}(t_{(k,k^\star),j,y})=(\omega_{l,j}(t_{(k,k^\star),j,y}),\omega_{u,j}(t_{(k,k^\star),j,y}))$ where $\omega_{l,j}(t_{(k,k^\star),j,y})$ is the lower boundary for testing hypothesis $H_{k.k^\star}$ at stage $j$ and $\omega_{u,j}(t_{(k,k^\star),j,y})$ is the upper boundary, for $k,k^\star=1,\hdots K$ and $j=1,\hdots, J$, with $t_{(k,k^\star),j,y}=a_1,a_2,\hdots,a_8$ that defines the values of $\omega_{j}(t_{(k,k^\star),j,y})$, which are defined in Table \ref{tab:tandavalues}.

\begin{table}[h]
\centering
 \caption{ \small The value of $\omega_{j}(t_{(k,k^\star),j,y})$, where $\omega_{j}(t_{(k,k^\star),j,y})= \{\omega_{l,j}(t_{(k,k^\star),j,y}), \omega_{u,j}(t_{(k,k^\star),j,y}) \}$, for given stage $j=1,\hdots,J$ depending on the integer value of  $t_{(k,k^\star),j,y}$.}

\resizebox{1.00\textwidth}{!}{\begin{tabular}{c c c c c c c c c c}
$t_{(k,k^\star),j,y}$  & \vline & $a_1$ & $a_2$ & $a_3$ & $a_4$ & $a_5$ & $a_6$ & $a_7$ & $a_8$
\\
\hline
$\omega_{j}(t_{(k,k^\star),j,y})$ & \vline & $ \{-\infty, -u_j\} $  & $\{ -u_j,-u^\star_j \}$ & $ \{ -u^\star_j, u^\star_j \}$ & $\{ u^\star_j, u_j \}$ &  $\{ u_j, \infty \}$ & $\{ -\infty, \infty \}$ & $\{ -u_j, u_j \}$ & $\{ -u_j, u_j \}$
\end{tabular}}
\label{tab:tandavalues} 
\end{table}
\par 
  Notation $a_1, a_2, \hdots, a_5,$ are used to define the 5 possible areas in which each test statistic value could be, assuming it was still being tested in the trial. The test statistic is either: below the outer lower boundary ($a_1$); between the outer and inner lower boundaries ($a_2$); between the inner lower and upper boundaries ($a_3$); inner and outer upper boundaries ($a_4$); above the outer upper boundary ($a_5$). The remaining 3 values, $a_6, a_7, a_8,$ are used to simplify and streamline the calculations. The notation $a_6$ is used for a test statistic in which one of the treatments being tested has stopped the trial. One can remove any integrals for which $t_{(k,k^\star),j,y}=a_6$ as long as the corresponding $\boldsymbol{\theta}$ and $\Sigma$ values are also removed. 
The notation $a_7$ is used for a test statistic in which at least one of the treatments being tested is dropped at the current stage and the test statistic is not significant enough to cause another treatment to be dropped. Finally $a_8$ is used as for any stage in which $u^\star_j=0$, there are now only $3$ possible outcomes for each test statistic of interest, $a_1, a_5$ and $a_8$.
\par
Based on $a_1, \hdots, a_8$, in the Supporting Information Section 3 it is given how to determine $\boldsymbol{\Omega_p}=\{ \Omega_{p,1}, \hdots, \Omega_{p,Y} \}$ and the corresponding values of $t_{(k,k^\star),j,y}$ for each $\Omega_{p,y}$. To calculate these we use the algorithm given in Supporting Information Section  3. This algorithm runs by first starting with $\boldsymbol{\Omega}=\{ \Omega_{1}, \hdots, \Omega_{Y^\star} \}$, which is a inclusive list of boundaries for all the trial test statistics, of length $Y^\star$. It assumes even if a test statistic falls below the outer lower boundary, or above the outer upper boundary or all are within the inner boundaries, that the test statistic will continued to be studied. This list is then reduced and altered in order to both decrease the number of elements needing calculating for a more efficient calculation of power, and to only leave combinations of bounds which lead to only the one clinically relevant treatment being found. 
  
%
%
%

\par

\subsection{Expected sample size}
The expected sample size is defined as $E(N|\Theta)$ where  $\Theta$ is the effect of all the treatments, so $\Theta=\{ \psi_1,\psi_2,\hdots, \psi_K \}$.  The expected sample size can be found as
\begin{equation}
E(N|\Theta)=\sum_{y'=1}^Y N(\Omega_{E,y'})Q_{\Theta}(\Omega_{E,y'}),
\label{Eq:ESS}
\end{equation}
where $Y'$ is the number of outcomes of interest, $Q_{\Theta}(\Omega_{E,y'})$ is the probability for each outcome and $N(\Omega_{E,y'})$ is the total number of patients associated with each outcome. Also similar to power we define $\mathbf{\Omega_E}=\{ \Omega_{E,1}, \hdots, \Omega_{E,Y'} \}$ where $\Omega_{E,y'}$ is the set of boundaries for that given configuration $y'$. One can use the algorithm given in Supporting Information Section 3 to calculate $\mathbf{\Omega_E}$. In Supporting Information Section 4 the equations to calculate both $Q_{\Theta}(\Omega_{E,y'})$ and $N(\Omega_{E,y'})$ are given. One can  also use $N(\Omega_{E,y'})$ and $Q_{\Theta}(\Omega_{E,y'})$ to find the distribution of the sample size as done in \cite{GreenstreetPeter2021Ammp}.

\section{Numerical results}
\subsection{Setting}
\label{Sec:Motivating}
Below, we revisit the setting of the motivating sepsis trial discussed in \cite{MagaretAmalia2016Doam}.  The power is set to 90\%, double triangular stopping boundaries are used \citep{WhiteheadJ.1997TDaA} and we use the requirement of FWER at 5\% (two-sided). 
\par
Following Theorem \ref{theorem:nonbind} the FWER is controlled in the strong sense if designed under the global null hypothesis for non-binding stopping rules, while in Example \ref{Exmp:Sprime} it is shown that this also holds for binding rules when $\alpha=5\%$. Therefore, for this trial, for both binding and non-binding stopping rules, the FWER will be controlled in the strong sense.
\par
Using the equations given in the Supporting Information Section 1 the double triangular boundaries are found such that the FWER equals 5\%. 
The power under the LFC and expected sample size were calculated using Equation \eqref{Eq:power} and the Equation \eqref{Eq:ESS}, respectively. Both $\boldsymbol{\Omega_p}$ and  $\boldsymbol{\Omega_E}$ were found from $\boldsymbol{\Omega}$ using the algorithm given in Supporting Information Section 3, with $Y^\star=3.814 \times 10^{12}$, and this being reduced to $Y=2974$ and $Y'=25907$ by using this algorithm, where $Y^\star,Y$ and $Y'$ is the number of configurations in $\boldsymbol{\Omega}, \boldsymbol{\Omega_p}, \boldsymbol{\Omega_E}$, respectively. The calculations were carried out using R \citep{Rref} and the packages \texttt{mvtnorm} \citep{mvtnorm}, \texttt{gtools} \citep{warnes2021package}, \texttt{doParallel} \citep{doParallel} and \texttt{foreach} \citep{foreach}. Code is available at https://github.com/petergreenstreet/MAMSAP. 
\subsection{Alternative designs}
The first alternative design considered is to run each comparison as a separate trial while using the double triangular boundaries. For the 4 arm example this will involve running 6 separate trials each with 3 stages. Each one of these trials is designed to have power of 90\% and a two-sided type I error of 5\%. When just powering each individual trial the power is the probability that a clinically relevant treatment is found superior to the other treatment. Therefore this power is different to considering the power across the multiple trials. Across all the trials we define the power under the LFC as the probability of finding the clinically relevant treatment as superior in all the trials it is involved in.
\par
For the first alternative design the total type I error across all the separate trials will equal $1-(1-\alpha)^6$ as $\eta=6$. We also consider the second alternative design where separate trials will be used with the total type I error across all the trials equaling 5\%, so, the type I error for each trial is set to 0.85\%. For this second alternative design we will ensure that the power is controlled at 90\% under the LFC. The total type II error under the LFC across all the trials equals  $1-(1-\beta)^{4-1}$ as there are $K-1$ hypotheses which need to be rejected for there not to be an error. The adjusted power for each trial is therefore $96.5\%$. 
\par
The third alternative design is the method described in \cite{WhiteheadJohn2020Eote}. In \cite{WhiteheadJohn2020Eote} they describe the type I error of interest as the probability of the pairwise type I error for each comparison and the power is the probability that a treatment $k$ is found inferior compared to another treatment $k'$ given $\psi_{k'}-\psi_k=\theta'$. Their approach uses the same trial structure as the design discussed here, however does not account for any correlation between test statistics of different treatments, or the fact all remaining test statistics need to be within the inner boundaries for the trial to stop. 
This approach is presented with the type I error and power as defined in \cite{WhiteheadJohn2020Eote}.
\par
As the third alternative design does not account for the correlation between the test statistics of different treatments we also consider controlling the FWER and power across the entire trial using the Bonferroni correction \citep{bonferroni1936teoria}. This is the fourth alternative approach, in which the type I error for each comparison is set to $\alpha/6=0.083\%$ and the power for each comparison is set to $1-\beta/(4-1)=96.7\%$.

\subsection{Results}  
The sample size, stopping boundaries and the expected sample size, along with power and FWER for the different design options are given in Table \ref{tab:Comparisondesigns}, with all the results being calculated analytically. As expected, the proposed MAMSAP design has the desired FWER control of 5\% and power under LFC of 90\%. The maximum sample size for the design with binding boundaries is $243\times 4= 972$ while the design with non-binding boundaries has a maximum sample size of 984 patients. The expected sample size is studied under 4 configurations: The first is the global null hypothesis, $\Theta_0=(\psi,\psi,\psi,\psi)$ where $\psi$ is the treatment effect of a treatment without a clinically relevant effect; the second  is the LFC, $\Theta_1=(\psi+\theta',\psi,\psi,\psi)$; in the third configuration two treatments have a clinically relevant effect compared to the other treatments, $\Theta_2=(\psi+\theta',\psi+\theta',\psi,\psi)$;  and the fourth, $\Theta_3=(\psi+\theta',\psi+\theta',\psi+\theta',\psi)$ has three treatments with a clinically relevant effect compared to the remaining treatment. The expected sample size under these configuration ranges from 749.9 patients under the null hypothesis to 629.7 patients under $\Theta_2$ for the MAMSAP design for binding boundaries. For the MAMSAP design for non-binding boundaries the expected sample size ranges from 636.6 patients under $\Theta_2$ to $758.0$ patients under the global null hypothesis.
\par
For the MAMSAP design with non-binding boundaries, if the inner boundaries rules are strictly followed, then the FWER is 4.8\%, 
highlighting the small conservatism that can occur if the non-binding boundaries rules are followed. The necessary increase in the stopping boundaries resulting from the use of non-binding rules means that one additional patient per arm per stage is needed to achieve power above 90\%.

%
\par
The operating characteristics for the competing approaches are given in Table \ref{tab:Comparisondesigns} for binding boundaries. The Whitehead approach results in a smaller sample size compared to MAMSAP, however this approach does not control the FWER nor achieves power at the desired level. For example the maximum sample size drops by 38\% compared to MAMSAP, at the cost of an FWER inflation of over $16\%$ and a drop in power by $8.9\%$. When using the Whitehead design with a Bonferroni adjustment, so that the FWER and the power are now controlled, the bounds and sample size are conservative resulting in a larger maximum and expected sample size than required. As a result the expected sample size under the global null hypothesis has increased from 749.9 for the MAMSAP design to 820.1 for the Bonferroni adjusted Whitehead design.
\par
Table \ref{tab:Comparisondesigns} also shows the operating characteristics of running multiple separate trials. Even when not controlling for the FWER or power across all the trials there is still an increase in sample size compared to running MAMSAP due to the need to recruit each treatment group multiple times. The maximum sample size increases from 972 to 1800. Additionally the FWER is inflated to 26.5\% and the power under the LFC is only 73.6\%. The increase in sample size is further emphasised when the power and FWER are controlled across all the trials at the desired level. Now the maximum sample size is increased by over $300\%$ to $3204$ compared to the MAMSAP design.   
\par
In the Supporting Information Section 5 the results for the competing approaches when using non-binding stopping boundaries are provided. Additionally in the Supporting Information Section 6 the probability of finishing the trial with $i$ out of $K'$ clinically relevant treatments is shown for the MAMSAP design under both binding and non-binding stopping rules. In the Supporting Information Section 7 an alternative design of running a set of sequential separate trials is given, where the next trial is based on the previous trial.

\begin{table}[]
\centering
 \caption{\small Operating characteristics of the MAMSAP design for both binding and non-binding boundaries along with the operating characteristics of the competing approaches for binding stopping boundaries.}

\resizebox{0.75\textwidth}{!}{\begin{tabular}{c|c|c|c|c|c|c}
\multirow{4}{*}{Design} & \multirow{4}{*}{$\begin{pmatrix}
u_1\\ 
u_2\\ 
u_3
\end{pmatrix}$}  & \multirow{4}{*}{$\begin{pmatrix}
u_1^\star\\ 
u_2^\star\\ 
u_3^\star
\end{pmatrix}$}  &  \multirow{4}{*}{$\begin{matrix}
\text{FWER}\\ 
\text{Power}\\ 
\end{matrix}$}  & \multirow{4}{*}{$\begin{pmatrix}
n_1\\ 
n_2\\ 
n_3
\end{pmatrix}$}  & \multirow{4}{*}{$\begin{matrix}
\max(N)\\ 
\end{matrix}$}  & \multirow{4}{*}{$\begin{matrix}
E(N|\Theta_0) \\
E(N|\Theta_1) \\ 
E(N|\Theta_2) \\
E(N|\Theta_3)
\end{matrix}$}  \\
& & & & & &  \\
& & & & & &  \\
& & & & & &  \\
\hline
\multirow{2}{*}{MAMSAP} & \multirow{4}{*}{$\begin{pmatrix}
3.166\\ 
2.798\\ 
2.742
\end{pmatrix}$} & \multirow{4}{*}{$\begin{pmatrix}
0\\ 
1.679\\ 
2.742
\end{pmatrix}$} & \multirow{4}{*}{$\begin{matrix}
0.050\\ 
0.900
\end{matrix}$} & \multirow{4}{*}{$\begin{pmatrix}
81\\ 
162\\ 
243 
\end{pmatrix}$}&  \multirow{4}{*}{$\begin{matrix}
972\\ 
\end{matrix}$} & \multirow{4}{*}{$\begin{matrix}
749.9 \\
647.5\\ 
629.7\\ 
669.9
\end{matrix}$} \\
\multirow{2}{*}{with binding}  & & & & & &  \\
\multirow{2}{*}{boundaries} & & & & & &  \\
& & & & & &  
\\
\hline
\multirow{2}{*}{MAMSAP with} & \multirow{4}{*}{$\begin{pmatrix}
3.181\\ 
2.811\\ 
2.755
\end{pmatrix}$}  & \multirow{4}{*}{$\begin{pmatrix}
0.000\\ 
1.687\\ 
2.755
\end{pmatrix}$}  & \multirow{4}{*}{$\begin{matrix}
0.048\\ 
0.903
\end{matrix}$}  & \multirow{4}{*}{$\begin{pmatrix}
82\\ 
164\\ 
246
\end{pmatrix}$}  &  \multirow{4}{*}{$\begin{matrix}
984\\  
\end{matrix}$}  & \multirow{4}{*}{$\begin{matrix}
758.0\\
654.5\\ 
636.6\\ 
677.2
\end{matrix}$}  \\
\multirow{2}{*}{non-binding} & & & & & & \\
\multirow{2}{*}{boundaries} & & & & & &  \\
 & & & & & &  \\
\hline
 & \multirow{4}{*}{$\begin{pmatrix}
2.484\\ 
2.195\\ 
2.151
\end{pmatrix}$} & \multirow{4}{*}{$\begin{pmatrix}
0\\ 
1.317\\ 
2.151
\end{pmatrix}$}  &\multirow{4}{*}{$\begin{matrix}
0.213\\ 
0.811
\end{matrix}$}  & \multirow{4}{*}{$\begin{pmatrix}
50\\ 
100\\ 
150
\end{pmatrix}$} &\multirow{4}{*}{$\begin{matrix}
600\\ 
\end{matrix}$}  &\multirow{4}{*}{$\begin{matrix}
488.8 \\
397.6\\ 
393.6\\ 
428.7
\end{matrix}$}    \\
 \multirow{1}{*}{Whitehead}   & & & & & &  \\
\multirow{1}{*}{design} & & & & & &  \\
& & & & & &  \\
\hline 
\multirow{1}{*}{ Bonferroni} &\multirow{4}{*}{$\begin{pmatrix}
3.213\\ 
2.840\\ 
2.783
\end{pmatrix}$} &\multirow{4}{*}{$\begin{pmatrix}
0\\ 
1.704\\ 
2.783
\end{pmatrix}$}   &\multirow{4}{*}{$\begin{matrix}
0.045\\ 
0.929
\end{matrix}$}  & \multirow{4}{*}{$\begin{pmatrix}
89 \\ 
178\\ 
267
\end{pmatrix}$} & \multirow{4}{*}{$\begin{matrix}
1068 
\end{matrix}$} & \multirow{4}{*}{$\begin{matrix}
820.1 \\
689.9\\ 
676.4\\ 
726.6
\end{matrix}$}  \\
\multirow{1}{*}{adjusted} & & & & & &  \\
\multirow{1}{*}{Whitehead} & & & & & &  \\
\multirow{1}{*}{design} & & & & & &  \\
\hline
  & \multirow{4}{*}{$\begin{pmatrix}
2.484\\ 
2.195\\ 
2.151
\end{pmatrix}$} & \multirow{4}{*}{$\begin{pmatrix}
0\\ 
1.317\\ 
2.151
\end{pmatrix}$}  & \multirow{4}{*}{$\begin{matrix}
0.265\\ 
0.736
\end{matrix}$} & \multirow{4}{*}{$\begin{pmatrix}
50\\ 
100\\ 
150
\end{pmatrix}$} & \multirow{4}{*}{$\begin{matrix}
1800\\ 
\end{matrix}$} & \multirow{4}{*}{$\begin{matrix}
1284.5 \\
1199.3\\ 
1170.8\\ 
1199.3
\end{matrix}$}  \\
\multirow{1}{*}{Separate}  & & & & & &  \\
\multirow{1}{*}{trials} & & & & & &  \\
 & & & & & &  \\
\hline
\multirow{1}{*}{FWER}  & \multirow{4}{*}{$\begin{pmatrix}
3.205\\ 
2.833\\ 
2.776
\end{pmatrix}$} & \multirow{4}{*}{$\begin{pmatrix}
0\\ 
1.699\\ 
2.776
\end{pmatrix}$}  & \multirow{4}{*}{$\begin{matrix}
0.050\\ 
0.905
\end{matrix}$} & \multirow{4}{*}{$\begin{pmatrix}
89\\ 
178\\ 
267
\end{pmatrix}$} & \multirow{4}{*}{$\begin{matrix}
3204\\ 
\end{matrix}$} & \multirow{4}{*}{$\begin{matrix}
2223.8 \\
2090.4\\ 
2045.9\\ 
2090.4
\end{matrix}$}   \\
\multirow{1}{*}{controlled} & & & & & &  \\
\multirow{1}{*}{separate} & & & & & &  \\
\multirow{1}{*}{trials} & & & & & &  \\
\end{tabular}}
\label{tab:Comparisondesigns} 
\end{table}

\section{Discussion}
\label{Sec:Discussion}
\par
The work presented here allows for the design of multi-arm multi-stage trials in which there is no control treatment so all pairwise comparisons are conducted. We have developed a method which allows the calculation of both binding and non-binding stopping boundaries to control the FWER under the global null hypothesis. We show that the design controls the FWER in the strong sense when non-binding rules are used and a test with a finite number of comparisons has been developed in order to test if the FWER is controlled in the strong sense for binding boundaries. Furthermore expressions for the power under the LFC and the expected sample size are provided. Based on a  motivating example we show that the proposed method, MAMSAP design, outperforms alternative approaches that also control FWER and power.
\par
\par
In the Supporting Information Section 8 it is shown that the FWER holds in the strong sense for the double triangular boundaries, for equal sample size per arm per stage, for up to an 8 arm 15 stage example with FWER set to 2.5\%, 5\%, and 10\%. Beyond 8 arms and 15 stages the computation becomes too slow and unstable to accurately check Theorem \ref{theorem:FWERshow}. One could therefore extend this work to see if there is a way to prove strong control of FWER for the double triangular stopping boundaries or if there is a counter example.
\par
Building on the work on adding arms \citep{BurnettThomas2020Aeta,GreenstreetPeter2021Ammp, greenstreet2023preplanned} in the controlled setting, future work will consider this problem for the all pairwise setting.  Such an extension raises questions around the use of non-concurrent treatments and potential bias caused by time trends. 
Both of which are well studied when there is a common control  \citep{lee2020including, MarschnerIanC2022Aoap}.
\par
\par
This paper introduces a framework for designing multi-arm multi-stage trials in which there is no control treatment, centred around normal continuous endpoints. Using the methodology proposed by \cite{JakiT2013Coca}, this approach can accommodate other endpoints, including binary, as discussed in \cite{MagaretAmalia2016Doam}. As one employs this methodology, it is essential to acknowledge potential computational challenges related to the computation of high-dimensional multivariate normal distributions and the large number of feasible outcomes of the trial, particularly when dealing with large numbers of arms and stages. If such challenges arise, one may consider approaches outlined in \cite{BlondellLucy2021Game} for handling the high-dimensional multivariate normal distributions, or alternatively, use a simulation-based approach.

\section*{Acknowledgements}
This report is independent research supported by the National Institute for Health Research (NIHR300576). The views expressed in this publication are those of the authors and not necessarily those of the NHS, the National Institute for Health Research or the Department of Health and Social Care (DHSC). TJ and PM also received funding from UK Medical Research Council (MC UU 00002/14, MC UU 00002/19, MC\_UU\_00040/03). This paper is based on work completed while PG was part of the EPSRC funded STOR-i centre for doctoral training (EP/S022252/1). PG is supported by a CANSTAT trainee award funded by CIHR grant \#262556. For the purpose of open access, the author has applied a Creative Commons Attribution (CC BY) licence to any Author Accepted Manuscript version arising.

\bibliographystyle{chicago}

\bibliography{Bibliography-MM-MC}

\newpage
\begin{center}
    \LARGE{\textbf{Supporting Information}}
\end{center}

\beginsupplement
\section{Calculating FWER under the global null}
\label{SI:CalFWERglobalnull}
When calculating $P(R_{\mathbf{G},\mathbf{G}})$ there are, at each stage, only two possibilities that need to be calculated, either all arms are between $-u$ and $u$ or all are between $-u^\star$ and $u^\star$. Therefore we define $\bar{U}_{j}(\cdot)$ where $\bar{U}_{j}(1)=u^\star_j$ and $\bar{U}_{j}(0)=u_j$. So 
\begin{equation}
P(R_{\mathbf{G},\mathbf{G}})=\sum^J_{j=1} \sum_{\substack{q_j=1 \& q_i \in \{0,1\} \\ i=1,2,\ldots,j}} -1^{(\sum_{i=1}^j (q_i)-1)} \int^{\bar{U}_{1}(q_1)}_{-\bar{U}_{1}(q_1)}\hdots \int^{\bar{U}_{1}(q_1)}_{-\bar{U}_{1}(q_1)}  \hdots \int^{\bar{U}_{j}(q_j)}_{-\bar{U}_{j}(q_j)} \hdots \int^{\bar{U}_{j}(q_j)}_{-\bar{U}_{j}(q_j)} \phi(\mathbf{z},\mathbf{0},\Sigma_{[1:\eta j]})  d\mathbf{z},
\end{equation}
where $\phi(\mathbf{z},\boldsymbol{\mu},\Sigma)$ is the probability density function of a multivariate normal distribution with mean $\boldsymbol{\mu}$ and covariance matrix $\Sigma$. The equation for the covariance matrix, $\Sigma$, is defined below and $[\cdot]$ defines the rows and columns of the covariance matrix, needed. 
For non-binding boundaries one can find $P(R'_\mathbf{G})$ as
\begin{equation}
\int^{u_1}_{-u_1}\hdots \int^{u_1}_{-u_1}  \hdots \int^{u_J}_{-u_J} \hdots \int^{u_J}_{-u_J} \phi(\mathbf{z},\mathbf{0},\Sigma)  d\mathbf{z}.
\label{eq:NonbindFWER}
\end{equation}
To calculate in $P(R'_\mathbf{T^\star_\beta})$ one can use Equation \eqref{eq:NonbindFWER}, however, now simply excluding any test statistics which are related to treatments which don't have equal treatment effect. 
\subsection{The correlation matrix equation}
\label{SI:corrmat}
The correlation matrix, $\Sigma$, structure is
\begin{align*}
\Sigma& =  \left(
  \begin{matrix}
\rho_{((1,2),1) ,((1,2),1)} & \rho_{((1,2),1),((1,3),1)} & \hdots 
\\
\rho_{((1,3),1),((1,2),1)} & \rho_{((1,3),1),((1,3),1)} & \hdots 
\\
\vdots & \vdots & \ddots 
\\
\rho_{((K-1,K),1),((1,2),1)} & \rho_{((1,K),1),((1,3),1)} & \hdots 
\\
\rho_{((1,2),2),((1,2),1)} & \rho_{((1,2),2),((1,3),1)} & \hdots 
\\
\vdots & \vdots & \ddots 
\\
\rho_{((K-1,K),J),((1,2),1)} & \rho_{((K-1,K),J),((1,3),1)} & \hdots \\
\end{matrix}\right.
\\ &
\left.
  \begin{matrix}
\rho_{((1,2),1),((K-1,K),1)} & \rho_{((1,2),1),((1,2),2)} & \hdots & \rho_{((1,2),1),((K-1,K),J)} \\
\\
 \rho_{((1,3),1),((K-1,K),1)} & \rho_{((1,3),1),((1,2),2)} & \hdots & \rho_{((1,3),1),((K-1,K),J)} 
\\
\vdots & \vdots & \ddots & \vdots
\\
 \rho_{((1,K),1),((K-1,K),1)} & \rho_{((1,K),1),((1,2),2)} & \hdots & \rho_{((1,K),1),((K-1,K),J)} 
\\
 \rho_{((1,2),2),((K-1,K),1)} & \rho_{((1,2),2),((1,2),2)} & \hdots & \rho_{((1,2),2),((K-1,K),J)} 
\\
 \vdots & \vdots & \ddots & \vdots
\\
 \rho_{((K-1,K),J),((K-1,K),1)} & \rho_{((K-1,K),J),((1,2),2)} & \hdots & \rho_{((K-1,K),J),((K-1,K),J)} \\
\end{matrix}\right),
\end{align*}
where 
\begin{equation*}
\rho_{((k_1,k_1^\star),j),((k_2,k_2^\star),j^\star)}=corr(Z_{(k_1,k_1^\star),j},Z_{(k_2,k_2^\star),j^\star}).
\end{equation*}

\section{Example of the FWER not being controlled in the strong sense under the global null}

\begin{exmp}
Consider a trial design with 3 arms and 2 stages with equal number of patients per arm per stage. If  $u_1=\infty$ and $u_1^\star=2.2$ then the final boundary needs to be $u_2=1.558$ to control the FWER under the global null hypothesis at a 2-sided level of 5\%. If there are 10 patients per arm per stage, then if $\psi_1+5=\psi_2=\psi_3$ then the FWER under this configuration is $11.9\%$. This is because when $\psi_1+5=\psi_2=\psi_3$ the trial will almost never stop at the first stage, as $Z_{(1,2),1}$ and $Z_{(1,3),1}$ will be, with high probability, less than $-2.2$, and it is not possible to drop a treatment for being inferior as  $u_1=\infty$,  therefore at the final stage the probability of declaring that $\psi_2\neq\psi_3$ is $11.9\%$ with the boundary of $u_2=1.558$.
\hfill $\blacksquare$
\end{exmp}

\section{Calculation of $\boldsymbol{\Omega_p}$ and $\boldsymbol{\Omega_E}$}
\label{App:CalcuationOmega}

Algorithm \ref{Alg:OmegapandE} starts with a $\boldsymbol{\Omega}$ which is an inclusive list of boundaries for all the trial test statistics, assuming that even if a test statistic falls below the outer lower boundary, or above the outer upper boundary or all are within the inner boundaries, that the test statistic will continued to be studied. Similar to $\boldsymbol{\Omega_p}$, $\boldsymbol{\Omega}= \{\Omega_{1}, \hdots, \Omega_{Y^\star}\}$ where each $\Omega_{y^\star}$ is the set of upper and lower boundaries required for that given configuration. The number of configurations is denoted $Y^\star$, so $\Omega_{y^\star}=\{\omega_{1}(t_{(1,1),1,y^\star}), \hdots, \omega_{J}(t_{(K-1,K),J,y^\star})\}$ for $y^\star=1,\hdots, Y^\star$. In total $\boldsymbol{\Omega}$ begins with a list of length $Y^\star=5^{J\eta}$ as every configuration of $t_{k,k^\star,j,y^\star}=a_1,\hdots,a_5$ is considered for every $k \neq k^\star$, $k, k^\star= 1,\hdots, K$, $j=1,\hdots,J$ and $y^\star=1,\hdots,Y^\star$. So we are testing $-\infty <Z_{k,k^\star,j}< l$; $l <Z_{k,k^\star,j}< -u^\star$; $-u^\star < Z_{k,k^\star,j}< u^\star$; $u^\star <Z_{k,k^\star,j}< u$; $u <Z_{k,k^\star,j}< \infty$  for every $Z_{k,k^\star,j}$.

\begin{algorithm}[H]
\caption{To find $\boldsymbol{\Omega_p}$ and $\boldsymbol{\Omega_E}$}
\begin{itemize}
\item[1] Generating every possible combination of $a_1,\hdots,a_5$ for every $t_{(k,k^\star),j,y^\star}$, where $y^\star=1,\hdots ,Y^\star$ where $Y^\star=5^{\eta J}$ . To create a set of all outcomes $\boldsymbol{\Omega}$ 
\item[2] Use Reduction 1 to remove any impossible sets of $\boldsymbol{\Omega}$.
\item[3] Use Reduction 2 to change for any stage in which $u^\star=0$ to replace the any $ t_{(k,k^\star),j,y^\star}= a_2,a_3,a_4$ with the values $t_{(k,k^\star),j,y^\star}=a_8$ then remove any duplicates sets in $\boldsymbol{\Omega}$.
\item[4] Use Reduction 3 to change the final stage to remove the any sets in $\boldsymbol{\Omega}$  with the $t_{(k,k^\star),J,y^\star} = a_2, a_4$. 
\item[5] Repeat the following steps for $j$ from $1:J$. 
\begin{itemize}
\item[i] If $j>1$ use Reduction 5 to replace any hypotheses which stopped the stage before with $t_{(k,k^\star),j,y^\star}=6$ and remove any duplicates sets in $\boldsymbol{\Omega}$.
\item[ii] Use Reduction 4 for stage $j$ to replace any $t_{(k,k^\star),j,y^\star}=a_2,a_3,a_4,a_8$ of treatments which stop at stage $j$  with $t_{(k,k^\star),j,y^\star}=a_7$ and remove any duplicates sets.
\end{itemize}
Now $\boldsymbol{\Omega_E}$ equals the reduced $\boldsymbol{\Omega}$.
\item[6] Use Reduction 6 to remove all sets of $\boldsymbol{\Omega}$ in which any $t_{(k',k^\star),j,y^\star}=a_1$ or $t_{(k,k'),j,y^\star}=a_5$ for hypothesis testing treatment $k'$.
\item[7] Use Reduction 7 to remove all sets of $\boldsymbol{\Omega}$ in which any $t_{(k',k^\star),J,y^\star}=a_1,a_2,a_3,a_4$ and $t_{(k,k'),J,y^\star}=a_2,a_3,a_4,a_5$ for hypothesis testing treatment $k'$.
\item[8] Use Reduction 8 to remove all sets of $\boldsymbol{\Omega}$ in which for each $j$ all $t_{(k,k^\star),j,y^\star}=a_1,a_3,a_5,a_6,a_7$ and at least one of $t_{(k,k^\star),j,y^\star}=a_3$. Now $\boldsymbol{\Omega_p}$ equals the reduced $\boldsymbol{\Omega}$.
\end{itemize}
\label{Alg:OmegapandE}
\end{algorithm}
Once the list $\boldsymbol{\Omega}$ has been created, Algorithm \ref{Alg:OmegapandE} is used to reduce the list to find $\boldsymbol{\Omega_E}$, using the following first 5 reductions and find $\boldsymbol{\Omega_p}$, using the following 8 reductions.
\par           
\underline{Reduction 1:} Test which of the 5 outcomes are possible for a particular $Z_{(k,k^\star),j}$ based on the outcomes of the other test statistics at stage $j$. This is because $Z_{(k,k^\star),j}$ can be rewritten in terms of $Z_{(\dot{k},k),j}$ and $Z_{(\dot{k},k^\star),j}$ where $\dot{k} <k < k^\star$ as
\begin{equation}
Z_{(k,k^\star),j} = \frac{Z_{(\dot{k},k^\star),j}\sqrt{V_{\dot{k},k^\star,j}}-Z_{(\dot{k},k),j}\sqrt{V_{\dot{k},k,j}}}{\sqrt{V_{k,k^\star,j}}}.
\end{equation}
Therefore the maximum value $Z_{(k,k^\star),j}$ for a given $\dot{k}$ is 
\begin{equation*}
\max(Z_{(k,k^\star),j}|\dot{k})=\frac{\max(Z_{(\dot{k},k^\star),j})\sqrt{V_{\dot{k},k^\star,j}}-\min(Z_{(\dot{k},k),j})\sqrt{V_{\dot{k},k,j}}}{\sqrt{V_{k,k^\star,j}}}.
\end{equation*}
Given all values of $\dot{k}$ which are smaller then $k$ then
\begin{equation*}
\max(Z_{(k,k^\star),j})=\min(\max(Z_{(k,k^\star),j}|1), \hdots \max(Z_{(k,k^\star),j}|k-1)).
\end{equation*}
Similarly the minimum value $Z_{(k,k^\star),j}$ for a given $\dot{k}$ is 
\begin{equation*}
\min(Z_{(k,k^\star),j}|\dot{k})=\frac{\min(Z_{(\dot{k},k^\star),j})\sqrt{V_{\dot{k},k^\star,j}}-\max(Z_{(\dot{k},k),j})\sqrt{V_{\dot{k},k,j}}}{\sqrt{V_{k,k^\star,j}}}.
\end{equation*}
Given all values of $\dot{k}$ which are smaller then $k$ gives 
\begin{equation*}
\min(Z_{(k,k^\star),j})=\max(\min(Z_{(k,k^\star),j}|1), \hdots \min(Z_{(k,k^\star),j}|k-1)).
\end{equation*}
Using the maximum value and minimum value that each $Z_{(k,k^\star),j}$ can take, given the range of values $Z_{(\dot{k},k),j}$ and $Z_{(\dot{k},k^\star),j}$ can take, results in a reduction in which of $a_1,\hdots, a_5$ need to be considered as the limits of $Z_{(k,k^\star),j}$. For example in a 3 arm case, with equal sample size per arm, if $t_{(1,3),j}=a_1$ (so $-\infty<Z_{(1,3),j}<-u_j$) and $t_{(1,2),j}=a_5$ ($u_j<Z_{(1,2),j}<\infty$) then we know that $Z_{(2,3),j}<-2u_j$ therefore the only possible area of $t_{(2,3),j}$ is $a_1$. 
\par
\underline{Reduction 2}: For any stage in which $u^\star_j=-u^\star_j=0$ there are only $3$ possible outcomes for each test statistic. 
\par
\underline{Reduction 3}: At the final stage where $u_J^\star=u_J$ there are only 3 outcomes: $-\infty <Z_{k,k^\star,J}< -u_J$; $-u^\star_J < Z_{k,k^\star,J}< u^\star_J$; $u_J <Z_{k,k^\star,J}< \infty$.
\par
\underline{Reduction 4}: If treatment $k$ is dropped at stage $j$ the remaining test statistics for treatment $k$ that are not significant to cause another treatment to be dropped, so between $-u_j$ and $u_j$, have no effect on the rest of the trial as treatment $k$ will be dropped from the following stage. Therefore for treatment $k$ which is dropped from the trial at a given stage $j$ the test statistics related to treatment $k$ have 3 outcomes of interest: $-\infty <Z_{k,k^\star,j}< -u_j$; $-u_j < Z_{k,k^\star,j}< u_j$; $u_j <Z_{k,k^\star,j}< \infty$. One is still interested in the area $-\infty <Z_{k,k^\star,j}< -u_j$ and $u_j <Z_{k,k^\star,j}< \infty$ as from the initial definition of $\Omega_y^\star$ for some $y^\star=1,\hdots,Y$ it is possible for example for $Z_{1,2,j}<-u_j$, $-u_j<Z_{1,3,j}<u_j$, $Z_{2,3,j}<-u_j$ even though this is not possible in reality, so this event will have probability 0 which needs to be accounted for.  
\par
\underline{Reduction 5}: If a treatment has already been dropped, then for the remaining stages the value of its test statistics no longer matter, as in the trial these test statistics would no longer be tested. Therefore for computational convenience  $-\infty < Z_{k,k^\star,j}< \infty$ if treatment $k$ or $k^\star$ was dropped from the trial at stage $j^\star$ where $j^\star<j$ as in reality this test statistic would no longer be of interest in a real trial.
\par
Furthermore when calculating the power under the LFC there are three further reductions that can be made which result in only treatment $k'$ being found as the clinically relevant treatment. 
\par
\underline{Reduction 6}: If treatment $k'$ is found to be the clinically relevant treatment then it can never have been dropped from the trial, therefore $-\infty <Z_{k',k^\star,j}< -u_j$ and $u_j <Z_{k,k',j}< \infty$ are not possible for test statistics still being tested at stage $j$.
\par
\underline{Reduction 7}: At the final stage any remaining treatments must be found inferior to treatment $k'$, therefore, $u_J <Z_{k',k^\star,J}< \infty$ and $-\infty <Z_{k,k',J}< -u_J$ for any treatments still being tested. 
\par
\underline{Reduction 8}: The trial can not stop early for all the treatments being found similar as this means that treatment $k'$ was not found superior to at least one treatment. Therefore one can remove all outcomes which have all remaining test statistics, at any stage $j$, falling within $-u^\star_j$ to $u^\star_j$.
\par
Using these 8 reductions as detailed in Algorithm \ref{Alg:OmegapandE} one can find $\boldsymbol{\Omega_p}$ and $\boldsymbol{\Omega_E}$.

\section{Calculation of $Q_{\Theta}(\Omega_{E,y'})$ and $N(\Omega_{E,y'})$}
\label{App:Expectation}
One can use Algorithm \ref{Alg:OmegapandE} given in Appendix \ref{App:CalcuationOmega} to calculate $\mathbf{\Omega_E}$. Now one can find the probability for each outcome $\Omega_{E,y'}$ given $\Theta$ ($Q_{\Theta}(\Omega_{E,y'})$):
\begin{equation}
Q_{\Theta}(\Omega_{E,y'})= \int^{\omega_{u,1}(t_{(1,2),1,y'})}_{\omega_{l,1}(t_{(1,2),1,y'})} \hdots \int^{\omega_{u,J}(t_{(K-1,K),J,y'})}_{\omega_{l,J}(t_{(K-1,K),J,y'})} \phi(\mathbf{z},\boldsymbol{\theta},\Sigma)  d\mathbf{z},
\end{equation}
where  $\boldsymbol{\theta}$ has $\psi_1,\hdots,\psi_K$ of interest and $\Sigma$ is defined in the Supporting Information Section \ref{SI:corrmat}. 
One needs to find the total number of patients associated with each outcome, 
\begin{equation*}
N(\Omega_{E,y'})=\sum_{k=1}^K n_{k,\bar{j}_{k,y'}},
\end{equation*}
where
\begin{equation*}
\bar{j}_{k,y'}=\min_j( [t_{(k^\star,k),j,y}=a_6 \; \forall \; k^\star=1,\hdots,k-1  \cap t_{(k,k'),j,y}=a_6 \;\forall \; k'=k+1,\hdots,K] \cup [j-1=J])-1,
\end{equation*}
so $\bar{j}_{k,y'}$ gives the stage at which treatment $k$ stopped being recruited to, for configuration $y'$.

\section{Non-binding results}
\label{SI:NonBinding}
Table \ref{tab:ComparisondesignsNon} gives the operating characteristics of the competing approaches for non-binding stopping boundaries as done for binding stopping boundaries in Table 3.
\begin{table}[H]
\centering
 \caption{Operating characteristics of the MAMSAP design and competing approaches for non-binding stopping boundaries.}

\begin{tabular}{c|c|c|c|c|c|c}

\multirow{4}{*}{Design} & \multirow{4}{*}{$\begin{pmatrix}
u_1\\ 
u_2\\ 
u_3
\end{pmatrix}$}  & \multirow{4}{*}{$\begin{pmatrix}
u_1^\star\\ 
u_2^\star\\ 
u_3^\star
\end{pmatrix}$}  &  \multirow{4}{*}{$\begin{matrix}
\text{FWER}\\ 
\text{Power}\\ 
\end{matrix}$}  & \multirow{4}{*}{$\begin{pmatrix}
n_1\\ 
n_2\\ 
n_3
\end{pmatrix}$}  & \multirow{4}{*}{$\begin{matrix}
\max(N)\\ 
\end{matrix}$}  & \multirow{4}{*}{$\begin{matrix}
E(N|\Theta_0) \\
E(N|\Theta_1)\\ 
E(N|\Theta_2) \\
E(N|\Theta_3)
\end{matrix}$}  \\
& & & & & &  \\
& & & & & &  \\
& & & & & &  \\
\hline
 & \multirow{4}{*}{$\begin{pmatrix}
3.181\\ 
2.811\\ 
2.755
\end{pmatrix}$}  & \multirow{4}{*}{$\begin{pmatrix}
0.000\\ 
1.687\\ 
2.755
\end{pmatrix}$}  & \multirow{4}{*}{$\begin{matrix}
0.048\\ 
0.903
\end{matrix}$}  & \multirow{4}{*}{$\begin{pmatrix}
82\\ 
164\\ 
246
\end{pmatrix}$}  &  \multirow{4}{*}{$\begin{matrix}
984\\  
\end{matrix}$}  & \multirow{4}{*}{$\begin{matrix}
758.0\\
654.5\\ 
636.6\\ 
677.2
\end{matrix}$}  \\
 \multirow{1}{*}{MAMSAP} & & & & & & \\
\multirow{1}{*}{ design}& & & & & & \\
 & & & & & &  \\
\hline
 & \multirow{4}{*}{$\begin{pmatrix}
2.517\\ 
2.225\\ 
2.180
\end{pmatrix}$} & \multirow{4}{*}{$\begin{pmatrix}
0.000\\ 
1.335\\ 
2.180
\end{pmatrix}$}  & \multirow{4}{*}{$\begin{matrix}
0.201\\ 
0.813
\end{matrix}$} &\multirow{4}{*}{$\begin{pmatrix}
51\\ 
102\\ 
153
\end{pmatrix}$} & \multirow{4}{*}{$\begin{matrix}
612\\ 
\end{matrix}$} & \multirow{4}{*}{$\begin{matrix}
497.8\\
406.8\\ 
402.1\\ 
437.1
\end{matrix}$}    \\
 \multirow{1}{*}{Whitehead} & & & & & &  \\
 \multirow{1}{*}{ design}& & & & & &  \\
 & & & & & &  \\
\hline
  \multirow{1}{*}{Bonferroni} & \multirow{4}{*}{$\begin{pmatrix}
3.235\\ 
2.859\\ 
2.801
\end{pmatrix}$} & \multirow{4}{*}{$\begin{pmatrix}
0\\ 
1.715\\ 
2.801
\end{pmatrix}$} &\multirow{4}{*}{$\begin{matrix}
0.042\\ 
0.930
\end{matrix}$}  & \multirow{4}{*}{$\begin{pmatrix}
90\\ 
180\\ 
270
\end{pmatrix}$} & \multirow{4}{*}{$\begin{matrix}
1080\\ 
832.0
\end{matrix}$} & \multirow{4}{*}{$\begin{matrix}
698.2\\ 
684.1\\ 
734.0
\end{matrix}$} \\
 \multirow{1}{*}{adjusted} & & & & & & \\
 \multirow{1}{*}{Whitehead} & & & & & & \\
 \multirow{1}{*}{design} & & & & & & \\
\hline
 & \multirow{4}{*}{$\begin{pmatrix}
2.517\\ 
2.225\\ 
2.180
\end{pmatrix}$} & \multirow{4}{*}{$\begin{pmatrix}
0.000\\ 
1.335\\ 
2.180
\end{pmatrix}$} & \multirow{4}{*}{$\begin{matrix}
0.248\\ 
0.739
\end{matrix}$} & \multirow{4}{*}{$\begin{pmatrix}
51\\ 
102\\ 
153
\end{pmatrix}$} & \multirow{4}{*}{$\begin{matrix}
1836\\ 
\end{matrix}$} &\multirow{4}{*}{$\begin{matrix}
1308.9\\
1224.6\\ 
1196.5\\ 
1224.6
\end{matrix}$}  \\
\multirow{1}{*}{Separate} & & & & & &  \\
 \multirow{1}{*}{trials}& & & & & &  \\
 & & & & & &  \\
\hline
\multirow{1}{*}{FWER}   & \multirow{4}{*}{$\begin{pmatrix}
3.227\\ 
2.852\\ 
2.794
\end{pmatrix}$} & \multirow{4}{*}{$\begin{pmatrix}
0\\ 
1.711\\ 
2.794
\end{pmatrix}$} & \multirow{4}{*}{$\begin{matrix}
0.047\\ 
0.901
\end{matrix}$} & \multirow{4}{*}{$\begin{pmatrix}
89\\ 
178\\ 
267
\end{pmatrix}$} & \multirow{4}{*}{$\begin{matrix}
3204\\ 
\end{matrix}$} & \multirow{4}{*}{$\begin{matrix}
2222.0 \\
2095.7\\ 
2053.7\\ 
2095.7
\end{matrix}$}  \\
\multirow{1}{*}{controlled} & & & & & &\\
\multirow{1}{*}{separate}  & & & & & &\\
\multirow{1}{*}{trials} & & & & & & \\
\end{tabular}
\label{tab:ComparisondesignsNon} 
\end{table}

\section{The probability of finishing the trial with $i$ out of $K'$ clinically relevant treatments}
In Table \ref{tab:Breakdown} the probability of finishing the trial with $i$ out of $K'$ clinically relevant treatments is shown for the MAMSAP design under both binding and non-binding stopping rules. Here 
 $\Theta_4=(\psi+\theta',\psi+\theta',\psi+\theta',\psi+\theta')$. One should note that $\Theta_4$ is also equivalent to being under the global null as all the treatments have the same treatment effect. For both the binding and non-binding boundaries, under $\Theta_1$, the probability of finding one treatment with a clinically relevant effect equals the power under the LFC as planned. Moreover the probability of finding all 4 treatments with a clinically relevant effect equals one minus the FWER under $\Theta_4$. It can be seen for this example that when under the LFC the probability of finding $K'$ out of $K'$ clinically relevant treatments is at its lowest. It is at its highest when there are 2 clinically relevant treatments, with the probability of finding both clinically relevant treatments being 97.1\% and 97.2\% for binding and non-binding boundaries respectively. When there are two clinically relevant treatments then one or both of the two clinically relevant treatments can be found to be superior to the other null treatments. 
This is why the power under this configuration is higher compared to the LFC where there is only one clinically relevant treatment.
\par
On the right hand side of Table \ref{tab:Breakdown} there is the probability of ending the trial with $i^\star$ out of $K-K'$ treatments which do not have a clinically relevant effect. Under the global null hypothesis the trial will ideally finish with all 4 null treatments being declared similar. This is set to be controlled at the $5\%$ level, therefore for $i^\star=4$ in this case this gives $95\%$ for binding boundaries. When not all treatments are identified as equal under the global null hypothesis, most often only one treatment is dropped. For the binding boundaries under the LFC it can be seen that the probability of ending the trial with 1 null treatment is at 7.9\%, which is greater than the level of control for the FWER. This is because the power is set to 90\% so there is a 10\% chance that one or more of the null treatments will not have been rejected by the end of the trial.   
\par

In Table \ref{tab:Breakdown} the breakdown of the probabilities for the Whitehead design are also given. 
 For the Whitehead design for binding boundaries the effect of not controlling the FWER or power under the LFC across the entire design can be seen. Now there is only a 78.6\% chance of ending the trial without wrongly rejecting a null hypotheses as shown for $\Theta_0$. Additionally there is a 13.7\% chance that under the LFC there is still 1 treatment without a clinically relevant effect at the end of the trial. When studying the Bonferroni adjusted Whitehead design it can be seen that the design is overly conservative which is also shown in Table 3. When considering the separate trials design one is unable to produce these results as now there is a chance that the separate trials can end in  contradictory results. For example one may find that one can reject $H_{1,2}$ and declare that treatment 1 is superior so $\psi_1>\psi_2$, however one may find in another trial that $\psi_2>\psi_3$ and that $\psi_3>\psi_1$ as each trial is independent. As a result this is another drawback of running multiple separate trials. 

\par
In the Supporting Information Section \ref{SI:GeneralAlg} a more generalised algorithm of Algorithm \ref{Alg:OmegapandE} in Supporting Information Section \ref{App:CalcuationOmega} is given to find the set needed to calculate the power for $K'$ clinically relevant treatments.

\begin{table}[H]
\centering
 \caption{The probability of finishing the trial declaring $i$ out of $K'$ clinically relevant treatments under five different configurations: $\Theta_0=(\psi,\psi,\psi,\psi)$; $\Theta_1=(\psi+\theta',\psi,\psi,\psi)$; $\Theta_2=(\psi+\theta',\psi+\theta',\psi,\psi)$; $\Theta_3=(\psi+\theta',\psi+\theta',\psi+\theta',\psi)$; $\Theta_4=(\psi+\theta',\psi+\theta',\psi+\theta',\psi+\theta')$. Along with the probability of ending the trial with $i^\star$ treatments out of $K-K'$ which do not have a clinically relevant effect.}

\begin{tabular}{c|c|c|c|c|c|c|c|c}

\multicolumn{9}{c}{\textbf{Binding boundaries}}
\\
\hline
Treatment&\multicolumn{4}{c|}{Number of clinical relevant}&\multicolumn{4}{c}{Number of null}
\\
effect& 1  &  2 & 3 & 4 & 1 & 2 &  3   & 4 \\
\hline
$\Theta_0$& - & - & - & - &  0.000 & 0.004 & 0.045 & 0.950   \\
 $\Theta_1$ & 0.900  & -  & - & - & 0.079 & 0.016 & 0.004 & -   \\
$\Theta_2$ & 0.010 & 0.971 & - & - & 0.018 & 0.001 & - & -   \\
 $\Theta_3$ & 0.001 & 0.026 & 0.969 & - & 0.004 & - & - & -   \\
 $\Theta_4$ & 0.000 & 0.004 & 0.045 & 0.950 & - & - & - & -
\\

\hline
\multicolumn{9}{c}{\textbf{Non-Binding Boundaries}}
\\
\hline
Treatment&\multicolumn{4}{c|}{Number of clinical relevant}&\multicolumn{4}{c}{Number of null}
\\
effect& 1  &  2 & 3 & 4 & 1 & 2 &  3   & 4 \\
\hline
$\Theta_0$& - & - & - & - &  0.000 & 0.004 & 0.044 & 0.952   \\
 $\Theta_1$ & 0.903 & -  & - & - & 0.077 & 0.016 & 0.004 & -   \\
$\Theta_2$ & 0.009 & 0.972 & - & - & 0.017 & 0.001 & - & -   \\
 $\Theta_3$ & 0.001 & 0.025 & 0.970 & - &  0.004 & - & - & -   \\
 $\Theta_4$ & 0.000 & 0.004 & 0.044 & 0.952 & - & - & - & -   \\
 
 \hline

\multicolumn{9}{c}{\textbf{Whitehead Design}}
\\
\hline
Treatment&\multicolumn{4}{c|}{Number of clinical relevant}&\multicolumn{4}{c}{Number of null}
\\
effect& 1  &  2 & 3 & 4 & 1 & 2 &  3   & 4 \\
\hline
$\Theta_0$& - & - & - & - & 0.006 & 0.037 & 0.171 & 0.786 \\
 $\Theta_1$ & 0.811 & -  & - & - & 0.137 & 0.040 & 0.013 & -   \\
$\Theta_2$ & 0.051 & 0.899 & - & - & 0.047 & 0.003 & - & -   \\
 $\Theta_3$ & 0.014 & 0.113 & 0.860 & - & 0.013 & - & - & -   \\
 $\Theta_4$ & 0.006 & 0.037& 0.171 & 0.786 & - & - & - & -   \\ 
 
\hline

\multicolumn{9}{c}{\textbf{Bonferroni adjusted Whitehead Design}}
\\
\hline
Treatment&\multicolumn{4}{c|}{Number of clinical relevant}&\multicolumn{4}{c}{Number of null}
\\
effect& 1  &  2 & 3 & 4 & 1 & 2 &  3   & 4 \\
\hline
$\Theta_0$& - & - & - & - &  0.000 & 0.004 & 0.040 & 0.956   \\
 $\Theta_1$ & 0.929 & -  & - & - & 0.058 & 0.010 & 0.002 & -   \\
$\Theta_2$ & 0.009 & 0.980 & - & - & 0.011 & 0.000 & - & -   \\
 $\Theta_3$ & 0.001 & 0.023 & 0.974 & - & 0.002 & - & - & -   \\
 $\Theta_4$ & 0.000 & 0.004 & 0.040 & 0.956 & - & - & - & -   \\

\end{tabular}
\label{tab:Breakdown} 
\end{table}

\subsection{Generalised version of Algorithm S1}
\label{SI:GeneralAlg}
Let $\mathbf{k'}= \{k'_1,\hdots, k'_{K'}\}$ define the set of treatments with a clinically relevant effect. Let $\Omega_{p,K'}$ be the set of possible outcomes for power given that there are $K'$ clinically relevant treatments. Using Algorithm \ref{Alg:GenOmegap} the power for given $\mathbf{k'}$ can be found using $\boldsymbol{\Omega_{p,K'}}$ with Equation (3.4) with $\psi_1=\psi_2=\psi_{k_1'-1}=\psi_{k_1'}-\theta'=\psi_{k_1'+1}=\hdots = \psi_{k_{K'}'-1} =\psi_{k_{K'}'}-\theta'=\psi_{k_{K'}'+1}= \hdots =\psi_K$. 
For Algorithm \ref{Alg:GenOmegap} the final 3 reductions need to be edited and 1 additional one added compared to Algorithm \ref{Alg:OmegapandE}. 
\par
 \underline{Reduction $6^\star$:} Treatments $k_1',\hdots k'_{K'}$ can never be dropped from the trial therefore $-\infty <Z_{k_i',k^\star,j}< -u_j$ and $u_j <Z_{k,k_i',j}< \infty$ for all $k'_i=k_1',\hdots k'_{K'}$ are not possible for test statistics still being tested at stage $j$.
\par
\underline{Reduction $7^\star$:} At the final stage any remaining treatments not in the set $\mathbf{k'}$ must be found inferior to treatments in $\mathbf{k'}$ therefore $u_J <Z_{k_i',k^\star,J}< \infty$ and $-\infty <Z_{k,k_i',J}< -u_J$ for $k'_i \in \mathbf{k'}$ and $k,k^\star \notin \mathbf{k'}$ for any treatments still being tested.
\par 
\underline{Reduction $8^\star$:} The trial can not stop early for futility if any treatments $k  \notin \mathbf{k'}$ is still being tested. Therefore one can remove all outcomes which have all remaining test statistics, at any stage $j$, falling within $-u^\star_j$ to $u^\star_j$ which includes a treatment $k \notin \mathbf{k'}$. 
\par
\underline{Reduction $9^\star$:} If the trial stops at stage $j$ all the test statistics testing $k'_i \in \mathbf{k'}$ against $k'_{i^\star} \in \mathbf{k'}$ must finish falling within $-u^\star_j$ to $u^\star_j$.  

\begin{algorithm}[H]
\caption{To find $\boldsymbol{\Omega_{p,K'}}$}
\begin{itemize}
\item[1] Generating every possible combination of $a_1,\hdots,a_5$ for every $t_{(k,k^\star),j,y^\star}$, where $y^\star=1,\hdots ,Y^\star$ where $Y^\star=5^{\eta j}$ . To create a set of all outcomes $\boldsymbol{\Omega}$ 
\item[2] Use Reduction 1 to remove any impossible sets of $\boldsymbol{\Omega}$.
\item[3] Use Reduction 2 to change for any stage in which $u^\star=0$ to replace the any $ t_{(k,k^\star),j,y^\star}= a_2,a_3,a_4$ with the values $t_{(k,k^\star),j,y^\star}=a_8$ then remove any duplicates sets in $\boldsymbol{\Omega}$.
\item[4] Use Reduction 3 to change for the final stage to remove the any sets in $\boldsymbol{\Omega}$  with the $t_{(k,k^\star),J,y^\star} = a_2, a_4$. 
\item[5] Repeat the following steps for $j$ from $1:J$. 
\begin{itemize}
\item[i] If $j>1$ use Reduction 5 to replace any hypotheses which stopped the stage before with $t_{(k,k^\star),j,y^\star}=a_6$ and remove any duplicates sets in $\boldsymbol{\Omega}$.
\item[ii] Use Reduction 4 for stage $j$ to replace any $t_{(k,k^\star),j,y^\star}=a_2,a_3,a_4,a_8$ of treatments which stop at stage $j$ and remove any duplicates sets.
\end{itemize}
\item[6] Use Reduction $6^\star$ to remove all sets of $\boldsymbol{\Omega}$ in which any $t_{(k'_i,k^\star),j,y^\star}=a_1$ or $t_{(k,k'_i),j,y^\star}=a_5$ for hypothesis testing any treatment $k'_i \in \mathbf{k'}$.
\item[7] Use Reduction $7^\star$ to remove all sets of $\boldsymbol{\Omega}$ in which any $t_{(k'_i,k^\star),J,y^\star}=a_1,a_2,a_3,a_4$ and $t_{(k,k'_i),J,y^\star}=a_2,a_3,a_4,a_5$ for hypothesis testing treatment $k'_i \in \mathbf{k'}$ and $k,k^\star \notin \mathbf{k'}$.
\item[8] Use Reduction $8^\star$  to remove all sets of $\boldsymbol{\Omega}$ in which for each $j$ all $t_{(k,k^\star),j,y^\star}=a_1,a_3,a_5,a_6,a_7$ and at least one of $t_{(k,k^\star),j,y^\star}=a_3$ where either $k\notin \mathbf{k'}$ or $k^\star \notin \mathbf{k'}$. 
\item[9] Use Reduction $9^\star$  to remove all sets of $\boldsymbol{\Omega}$ in which for each $j$ all $t_{(k,k^\star),j,y^\star}=a_1,a_3,a_5,a_6,a_7$ and at least one of $t_{(k'_i,k'_{i^\star}),j,y^\star} \neq a_3$ where $k'_i,k'_{i^\star}\in \mathbf{k'}$. 
Now $\boldsymbol{\Omega_{p,K'}}$ equals the reduced $\boldsymbol{\Omega}$.
\end{itemize}
\label{Alg:GenOmegap}
\end{algorithm}

\section{Sequential separate trials}
\label{SI:SequentialST}
We consider the case that each of the separate trials are run sequentially for the binding boundary case. If one is testing treatment $k$ compared to treatment  $k^\star$ then one of three scenarios can happen which results in the trial stopping: 1) $Z_{(k,k^\star),j}>u_j$ then treatment $k^\star$ is found superior to treatment $k$ and only treatment $k^\star$ will continue being tested; 2) $Z_{(k,k^\star),j}<-u_j$ then treatment $k$ is found superior to treatment $k^\star$ and only treatment $k$ will continue being tested; 3) $-u_j^\star<Z_{(k,k^\star),j}<u_j^\star$ then treatment $k$ is found similar to treatment $k^\star$ and only one treatment will continue being tested, this will be chosen as treatment $k$. We assume in the model that treatment 1 will be compared to treatment 2 and then whichever treatment goes forward will be compared to treatment 3 and so on. Therefore now one only needs to conduct 3 trials for the motivating example.
\par
The FWER, power, and expected sample size under the least favourable configuration when using the separate trial design configuration given in Section 4.1 is shown in Table \ref{tab:Sequential}. As can be seen now the power under the LFC is now dependent on which treatment has the clinically relevant effect due to the ordering. For example for treatment 1 to be found clinically relevant it needs to be shown to be superior to all the other treatments, whereas for treatment 4 to be found clinically relevant it only need to be found superior to one treatment. 
\par
When comparing this design to MAMSAP it can be seen there is a decrease in sample size by using sequential separate trials with a decrease in maximum sample size of 72 patients. However this comes with a large decrease in the power, with it being as low as 73.6\% compared to the target of 90\% as well as in inflation of the FWER to 14.3\%.
\par
Additionally we consider the case of a sequence of separate trials are run where there is control of the FWER and power under the LFC. As 3 trials need to be conducted the type I error of each trial is $1-\sqrt[3]{1-0.05}=0.017$ and the power for each trial is $\sqrt[3]{0.90}=0.965$. This design is also given in Table \ref{tab:Sequential} with the stopping boundaries and sample size for each trial being
\begin{equation*}
\begin{pmatrix}
u_1\\ 
u_2\\ 
u_3
\end{pmatrix} = \begin{pmatrix}
2.958\\ 
2.615\\ 
2.562
\end{pmatrix},
\; \; \; \; \; \; 
\begin{pmatrix}
u_1^\star\\ 
u_2^\star\\ 
u_3^\star
\end{pmatrix} = \begin{pmatrix}
0.000\\ 
1.569\\ 
2.562
\end{pmatrix},
\; \; \; \; \; \;
\begin{pmatrix}
n_1\\ 
n_2\\ 
n_3
\end{pmatrix} =\begin{pmatrix}
81\\ 
162\\ 
243
\end{pmatrix}. 
\end{equation*}
As can be seen here this design requires a larger maximum sample size and expected sample size under the configuration considered here compared to the MAMSAP design, with the maximum sample size increased by 486 patients.
\begin{table}[H]
\centering
 \caption{Operating characteristics of the MAMSAP design for binding  boundaries along with the operating characteristics of the sequential separate trials for binding stopping boundaries.}

\begin{tabular}{c|c|c|c|c}
\multirow{4}{*}{Design} & \multirow{4}{*}{$\begin{matrix}
\text{FWER}\\ 
\end{matrix}$}  & \multirow{4}{*}{$\begin{matrix}
\text{Power }|(\psi+\theta',\psi,\psi,\psi)\\ 
\text{Power }|(\psi,\psi+\theta',\psi,\psi)\\
\text{Power }|(\psi,\psi,\psi+\theta',\psi)\\
\text{Power }|(\psi,\psi,\psi,\psi+\theta')\\
\end{matrix}$}  &   \multirow{4}{*}{$\begin{matrix}
\max(N)\\ 
E(N|(\psi,\psi,\psi,\psi))\\
\end{matrix}$}   & \multirow{4}{*}{$\begin{matrix}
E(N|(\psi+\theta',\psi,\psi,\psi))\\ 
E(N|(\psi,\psi+\theta',\psi,\psi))\\
E(N|(\psi,\psi,\psi+\theta',0))\\
E(N|(\psi,\psi,\psi,\psi+\theta'))\\
\end{matrix}$}  \\
& & & &   \\
& & & &   \\
& & & &   \\
\hline
\multirow{2}{*}{MAMSAP} & \multirow{4}{*}{$\begin{matrix}
0.050
\end{matrix}$} & \multirow{4}{*}{$\begin{matrix}
0.900\\
0.900\\
0.900\\
0.900
\end{matrix}$} & \multirow{4}{*}{$\begin{matrix}
972\\ 
749.9 \\
\end{matrix}$} & \multirow{4}{*}{$\begin{matrix}
647.5\\
647.5\\ 
647.5\\ 
647.5
\end{matrix}$} \\
\multirow{2}{*}{with binding}  & & & &   \\
\multirow{2}{*}{boundaries} & & & &    \\
& & & &  
\\
\hline
\multirow{2}{*}{Sequential}  & \multirow{4}{*}{$\begin{matrix}
0.143
\end{matrix}$} & \multirow{4}{*}{$\begin{matrix}
0.736 \\
0.736\\ 
0.815\\
0.903
\end{matrix}$}  & \multirow{4}{*}{$\begin{matrix}
900 \\
642.3
\end{matrix}$} & \multirow{4}{*}{$\begin{matrix}
557.0 \\
576.6 \\ 
588.2 \\
613.9
\end{matrix}$} \\
\multirow{2}{*}{separate}  & & & &   \\
\multirow{2}{*}{trials} & & & &  \\
 & & & &  \\
\hline
\multirow{1}{*}{FWER}  & \multirow{4}{*}{$\begin{matrix}
0.050
\end{matrix}$} & \multirow{4}{*}{$\begin{matrix}
0.902 \\
0.902\\ 
0.934\\
0.966
\end{matrix}$}  & \multirow{4}{*}{$\begin{matrix}
1458 \\
1022.1
\end{matrix}$} & \multirow{4}{*}{$\begin{matrix}
862.7 \\
874.9 \\ 
917.6 \\
969.0
\end{matrix}$}   \\
\multirow{1}{*}{controlled} & & & &   \\
\multirow{1}{*}{sequential} & & & &   \\
\multirow{1}{*}{separate trials} & & & &  \\
\end{tabular}
\label{tab:Sequential} 
\end{table}

\section{Double triangular boundaries}
\label{SI:Doubletri}

In Figure \ref{fig:comparisonplot} the double triangular stopping boundaries are found to control the FWER under the global null for binding boundaries. Here we consider an equal number of patients per stage per arm and the FWER control  target, $\alpha$ is 2.5\%, 5\% and 10\%. Figure 2 shows  $\max\bigg{(}1-P\bigg{(}\bigcap^J_{j=1} B_{S'_{i'},j} \bigg{)} \bigg{)}$ for all $S'_{i'} \in \mathbf{S'}$ for each  $\alpha$ level when using the boundaries found to control the FWER under the global null. It can be seen that, at all points in Figure \ref{fig:comparisonplot}, the probability of $\max\bigg{(}1-P\bigg{(}\bigcap^J_{j=1} B_{S'_{i'},j} \bigg{)} \bigg{)}$ is below that of the FWER of focus. Therefore by Theorem 3.3 this shows that for the double triangular stopping boundaries, with equal sample size per stage per arm, the FWER is controlled in the strong sense when using boundaries found under the global null hypothesis for up to 8 arms and 15 stages.

\begin{figure}[H]
\begin{subfigure}{1\textwidth}
  \centering
  \includegraphics[width=.75\linewidth,trim= 0 0.5cm 0 2cm, clip]{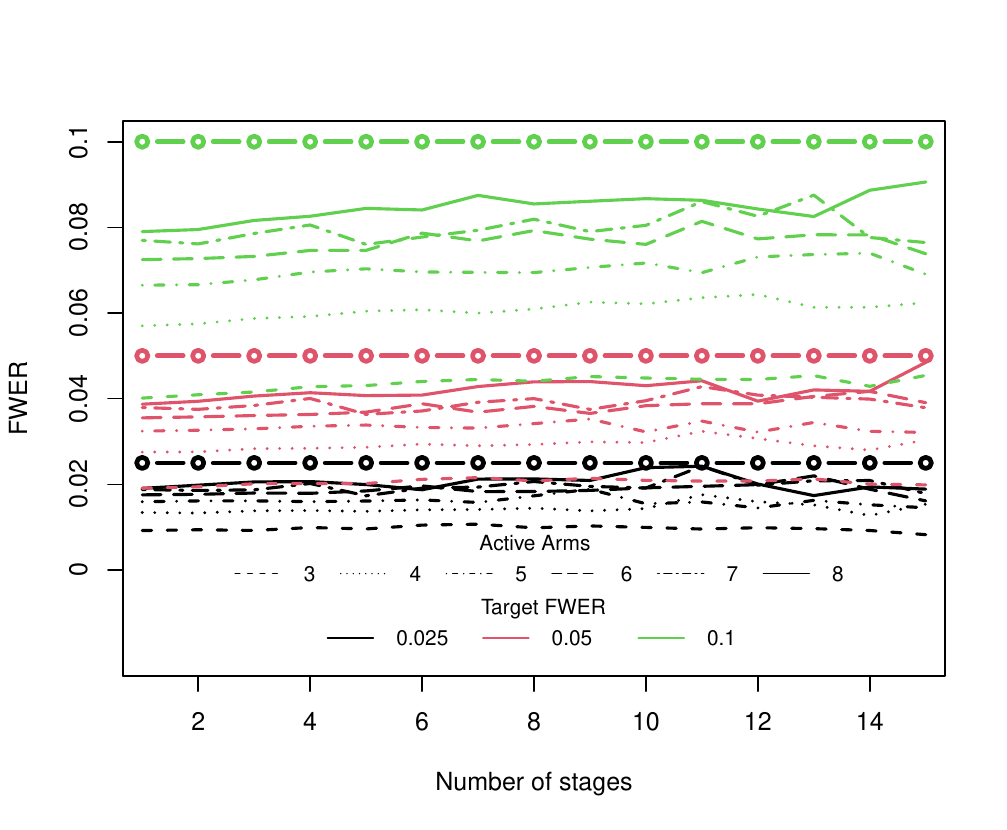}  
\end{subfigure}
\caption{Comparison of the $\max(1-P\bigg{(}\bigcap^J_{j=1} B_{S'_{i'},j} \bigg{)} )$ for all $S'_{i'} \in \mathbf{S'}$ with the desired FWER level of control, when using the binding double triangular stopping boundaries found under the global null.}
\label{fig:comparisonplot}
\end{figure}

\end{document}